\documentclass[letterpaper,11pt]{article}

\usepackage{newunicodechar}
\usepackage{silence}
\newunicodechar{♢}{\tikz \node[inner sep=1.5,draw,diamond] {};}
\newunicodechar{☆}{\tikz \node[inner sep=1,draw,star,star point ratio=2] {};}
\newunicodechar{△}{\triangle}
\newunicodechar{⬜}{\kern 0.5pt\tikz \node[inner sep=1.7,draw,regular polygon,regular polygon sides=4] {};\kern 0.5pt}
\newunicodechar{◯}{\tikz[baseline=-3pt] \node[inner sep=1.7,draw,cloud,cloud puffs=4,cloud puff arc=190] {};}

\newunicodechar{⊥}{\bot}
\newunicodechar{•}{\item}
\newunicodechar{✓}{\checkmark}
\newunicodechar{✗}{\xmark}
\newunicodechar{…}{\dots}
\newunicodechar{≔}{\coloneqq}
\newunicodechar{⁻}{^-}
\newunicodechar{⁺}{^+}
\newunicodechar{₋}{_-}
\newunicodechar{₊}{_+}
\newunicodechar{ℓ}{\ell}
\newunicodechar{•}{\item}
\newunicodechar{…}{\dots}
\newunicodechar{≔}{\coloneqq}
\newif\ifnormopen\normopenfalse
\newunicodechar{‖}{\ifnormopen\rVert\normopenfalse\else\rVert\normopentrue\fi}
\newunicodechar{≤}{\leq}
\newunicodechar{≥}{\geq}
\newunicodechar{≰}{\nleq}
\newunicodechar{≱}{\ngeq}
\newunicodechar{⊕}{\oplus}
\newunicodechar{⊗}{\otimes}
\newunicodechar{≠}{\neq}
\newunicodechar{¬}{\neg}
\newunicodechar{≡}{\equiv}
\newunicodechar{₀}{_0}
\newunicodechar{₁}{_1}
\newunicodechar{₂}{_2}
\newunicodechar{₃}{_3}
\newunicodechar{₄}{_4}
\newunicodechar{₅}{_5}
\newunicodechar{₆}{_6}
\newunicodechar{₇}{_7}
\newunicodechar{₈}{_8}
\newunicodechar{₉}{_9}
\newunicodechar{⁰}{^0}
\newunicodechar{¹}{^1}
\newunicodechar{²}{^2}
\newunicodechar{³}{^3}
\newunicodechar{⁴}{^4}
\newunicodechar{⁵}{^5}
\newunicodechar{⁶}{^6}
\newunicodechar{⁷}{^7}
\newunicodechar{⁸}{^8}
\newunicodechar{⁹}{^9}
\newunicodechar{∈}{\in}
\newunicodechar{∉}{\notin}
\newunicodechar{⊂}{\subset}
\newunicodechar{⊃}{\supset}
\newunicodechar{⊆}{\subseteq}
\newunicodechar{⊇}{\supseteq}
\newunicodechar{⊄}{\nsubset}
\newunicodechar{⊅}{\nsupset}
\newunicodechar{⊈}{\nsubseteq}
\newunicodechar{⊉}{\nsupseteq}
\newunicodechar{∪}{\cup}
\newunicodechar{∩}{\cap}
\newunicodechar{∀}{\forall}
\newunicodechar{∃}{\exists}
\newunicodechar{∄}{\nexists}
\newunicodechar{∨}{\vee}
\newunicodechar{∧}{\wedge}
\newunicodechar{ℝ}{\mathbb{R}}
\newunicodechar{ℕ}{\mathbb{N}}
\newunicodechar{𝔼}{\mathbb{E}}
\newunicodechar{𝔽}{\mathbb{F}}
\newunicodechar{ℤ}{\mathbb{Z}}
\newunicodechar{𝒪}{\mathcal{O}}
\newunicodechar{⌊}{\lfloor}
\newunicodechar{⌋}{\rfloor}
\newunicodechar{⌈}{\lceil}
\newunicodechar{⌉}{\rceil}
\newunicodechar{·}{\cdot}
\newunicodechar{∘}{\circ}
\newunicodechar{×}{\times}
\newunicodechar{↑}{\uparrow}
\newunicodechar{↓}{\downarrow}
\newunicodechar{→}{\rightarrow}
\newunicodechar{←}{\leftarrow}
\newunicodechar{⇒}{\Rightarrow}
\newunicodechar{⇐}{\Leftarrow}
\newunicodechar{↔}{\leftrightarrow}
\newunicodechar{⇔}{\Leftrightarrow}
\newunicodechar{↦}{\mapsto}
\newunicodechar{∅}{\varnothing}
\newunicodechar{∞}{\infty}
\newunicodechar{≅}{\cong}
\newunicodechar{≈}{\approx}
\newunicodechar{ℓ}{\ell}
\newunicodechar{𝟙}{\mathds{1}}
\newunicodechar{𝟘}{\mathds{0}}

\newunicodechar{α}{\alpha}
\newunicodechar{β}{\beta}
\newunicodechar{γ}{\gamma}
\newunicodechar{Γ}{\Gamma}
\newunicodechar{δ}{\delta}
\newunicodechar{Δ}{\Delta}
\newunicodechar{ε}{\varepsilon}
\newunicodechar{ζ}{\zeta}
\newunicodechar{η}{\eta}
\newunicodechar{θ}{\theta}
\newunicodechar{Θ}{\Theta}
\newunicodechar{ι}{\iota}
\newunicodechar{κ}{\kappa}
\newunicodechar{λ}{\lambda}
\newunicodechar{Λ}{\Lambda}
\newunicodechar{μ}{\mu}
\newunicodechar{ν}{\nu}
\newunicodechar{ξ}{\xi}
\newunicodechar{Ξ}{\Xi}
\newunicodechar{π}{\pi}
\newunicodechar{Π}{\Pi}
\newunicodechar{ρ}{\rho}
\newunicodechar{σ}{\sigma}
\newunicodechar{Σ}{\Sigma}
\newunicodechar{τ}{\tau}
\newunicodechar{υ}{\upsilon}
\newunicodechar{ϒ}{\Upsilon}
\newunicodechar{φ}{\varphi}
\newunicodechar{ϕ}{\phi}
\newunicodechar{Φ}{\Phi}
\newunicodechar{χ}{\chi}
\newunicodechar{ψ}{\psi}
\newunicodechar{Ψ}{\Psi}
\newunicodechar{ω}{\omega}
\newunicodechar{Ω}{\Omega}

\usepackage[letterpaper, margin=1in]{geometry}
\usepackage{amsmath,amsthm,amsfonts,amssymb}
\usepackage[utf8]{inputenc}
\pdfinterwordspaceon
\usepackage[margintag,small]{our-comments}
\usepackage{hyperref}
\usepackage[capitalise,noabbrev]{cleveref}
\usepackage{xspace}
\usepackage{bbm}
\usepackage{bm}
\usepackage{thm-restate}

\usepackage[most]{tcolorbox}
\tcbuselibrary{skins,breakable}

\usepackage{enumerate}
\usepackage[commentsnumbered,vlined]{algorithm2e}
\DontPrintSemicolon

\SetCommentSty{commentFont}
\SetKwProg{algo}{Algorithm}{:}{}
\usepackage{etoolbox}
\patchcmd{\SetProgSty}{ArgSty}{ProgSty}{}{}
\SetProgSty{op}

\usepackage{tikz}
\usepackage{booktabs}
\newcommand{\defn}[1]{\textbf{\emph{#1}}}

\newcommand{\prob}[1]{\Pr\left[#1\right]}
\DeclareMathOperator*{\E}{\mathbb{E}}

\newcommand{\ceiling}[1]{\left\lceil #1 \right\rceil}
\newcommand{\floor}[1]{\left\lfloor #1 \right\rfloor}
\newcommand{\norm}[1]{\lVert{#1}\rVert}

\newcommand{\random}{\textnormal{\textsc{Random}}\xspace}
\newcommand{\cluster}{\textnormal{\textsc{Cluster}}\xspace}
\newcommand{\bins}{\textnormal{\textsc{Bins}}\xspace}
\newcommand{\binsstar}{\textnormal{\textsc{Bins$^*$}}\xspace}
\newcommand{\clusterstar}{\textnormal{\textsc{Cluster$^*$}}\xspace}

\newcommand{\calA}{\mathcal{A}}
\newcommand{\B}{\mathcal{B}}
\newcommand{\calD}{\mathcal{D}}
\newcommand{\calE}{\mathcal{E}}

\newcommand{\calC}{\mathcal{C}}

\newcommand{\calU}{\mathcal{U}}
\newcommand{\clus}{\lambda}

\newcommand\numberthis{\addtocounter{equation}{1}\tag{\theequation}}

\newtheorem{theorem}{Theorem}
\newtheorem{lemma}[theorem]{Lemma}
\newtheorem{corollary}[theorem]{Corollary}

\newtheorem{claim}{Claim}

\def\restateThm#1#2{
    \marginpar{%
        \vspace{\baselineskip}%
        \centering%
        \footnotesize%
        \color{gray}%
        Originally stated on~page~\pageref{#1}%
    }%
    #2*
}

\addCommenter{guido}{Guido}{Guido}{red}
\addCommenter{stefan}{Stefan}{Stefan}{purple}
\addCommenter{peter}{Peter}{Peter}{blue}
\addCommenter{mfc}{MFC}{MFC}{green}

\def\refrel#1#2#3{\stackrel{\text{#2 \ref{#1}}}{#3}}

\def\eqrefrel#1#2#3{\stackrel{\text{#2(\ref{#1})}}{#3}}

\def\reasonrel#1#2{\stackrel{\text{#1}}{#2}}
\def\e{\mathrm{e}}

\begin{document}

\title{Optimal Uncoordinated Unique IDs}
\date{}

\author{
Peter C. Dillinger\thanks{Meta Platforms, Inc. \href{mailto:peterd@meta.com}{\texttt{peterd@meta.com}}}
\and
Mart\'{i}n Farach-Colton \thanks{Rutgers University. \href{mailto:martin@farach-colton.com}{\texttt{martin@farach-colton.com}}}
\and
Guido Tagliavini \thanks{Rutgers University. \href{mailto:guido.tag@rutgers.edu}{\texttt{guido.tag@rutgers.edu}}}
\and
Stefan Walzer \thanks{University of Cologne. \href{mailto:walzer@cs.uni-koeln.de}{\texttt{walzer@cs.uni-koeln.de}}}
}

\maketitle
\begin{abstract}
    In the \defn{Uncoordinated Unique Identifiers Problem} (UUIDP) there are $n$ independent instances of an algorithm $\calA$ that generates IDs from a universe $\{1, \dots, m\}$, and there is an adversary that requests IDs from these instances. The goal is to design $\calA$ such that it minimizes the probability that the same ID is ever generated twice across all instances, that is, minimizes the \defn{collision probability}. Crucially, no communication between the instances of $\calA$ is possible. Solutions to the UUIDP are often used as mechanisms for surrogate key generation in distributed databases and key-value stores. In spite of its practical relevance, we know of no prior theoretical work on the UUIDP.
    
    In this paper we initiate the systematic study of the UUIDP. We analyze both existing and novel algorithms for this problem, and evaluate their collision probability using worst-case analysis and competitive analysis, against oblivious and adaptive adversaries. In particular, we present an algorithm that is optimal in the \emph{worst case} against \emph{oblivious} adversaries, an algorithm that is at most a logarithmic factor away from optimal in the \emph{worst case} against \emph{adaptive} adversaries, and an algorithm that is optimal in the \emph{competitive} sense against both \emph{oblivious} and \emph{adaptive} adversaries.
\end{abstract}
\section{Introduction}

Unique identifiers (IDs) for data are essential to efficient data processing, storage, and retrieval. When natural keys for data do not exist or are impractical, surrogate keys are assigned and used as primary keys. In a distributed context, database systems including Cassandra, Microsoft's Transact-SQL, MongoDB, MySQL, Postgres and RocksDB ~\cite{MongoDBObjectID,MySQLUUID,PostgreSQLUUID,CassandraUUID,TransactSQLUUID,Dong2017RocksDBSpaceAmplification,Cao2020RocksDBWorkloads} generate unique surrogate keys \emph{without coordination} between nodes or use of a central authority. 

In the case of RocksDB, IDs are assigned to data files (also known as \textit{SSTs}) and data blocks within each file, for caching purposes. Although RocksDB is a single-instance key-value store, its users (e.g., Microsoft Bing's web data platform~\cite{Bing}, MyRocks~\cite{MyRocks} and ZippyDB~\cite{ZippyDB}) typically run multiple instances of it and distribute these instances across multiple nodes. Because data are moved between instances (e.g., to balance the load across the nodes), RocksDB's IDs have to be collision-free over \emph{all} instances, even though these instances are unaware of one another.

There are several benefits to generating IDs without coordination. First, coordination mechanisms are complex, and thus have significant engineering and operational costs. Second, implementations of coordination are often based on features such as MAC addresses (which are subject to spoofing) and clocks (which can become skewed), rendering them notoriously brittle. Finally, coordination in RocksDB is undesirable from a software design point of view, as it would force RocksDB instances to take responsibility for network awareness and communication.

Generating unique IDs without coordination may seem like fighting a losing battle, since communication is necessary to guarantee ID uniqueness. What we can aim for, however, is to generate IDs randomly, such that the probability that the same ID is ever generated by different instances is near zero.

Formally, in the \defn{Uncoordinated Unique Identifiers Problem} (UUIDP) an \defn{ID-generation algorithm} $\calA$ plays a game against an adversary. The game proceeds in steps, and at each step the adversary requests an ID from one of $n$ independent \defn{instances} of $\calA$. When an instance of $\calA$ receives a request, it must output a value from the space $[m] := \{1, \dots, m\}$ of IDs. Importantly, the instances cannot communicate with each other, and they do not know $n$ or the number of IDs that have been requested from other instances. \nocite{D2} The algorithm loses the game if the adversary causes a \defn{collision}, that is, if some ID is generated twice. The algorithm wins if by the time the game ends (for instance, after a fixed number of steps has been reached) a collision has not occurred. The goal is to design $\calA$ such that the probability that $\calA$ wins, called the \defn{collision probability}, is as small as possible.

Beyond database systems, uncoordinated generation of unique IDs is used in settings such as identification of network connections~\cite{nmcli}, identification of disk partitions~\cite{blkid}, and object identification in video games~\cite{Minecraft}. Such is the ubiquity of this problem that there exists a standardized solution to it~\cite{Leach2005RFC,ITU2004Recommendation}, known as \textit{globally unique identifiers} (GUIDs), and most programming languages have library implementations of GUIDs~\cite{GoUUID,BoostUUID,LinuxUUID,PythonUUID,JavaUUID,CSharpUUID}.

GUIDs are generated using some combination of a randomly generated integer and identifying metadata (e.g., creation timestamp or MAC address), with the particular combination depending on the GUID variant. This, however, presumes reliable metadata, which is often times impossible to ensure in practice---an adversary may tamper with the metadata to cause a collision. For this reason, we do not model metadata in the UUIDP. The the random part of GUIDs is modeled as the following simple algorithm for the UUIDP, that we call \random: every time an ID is requested, sample an integer from $[m]$ without replacement.

When a total of $d$ IDs are requested, the collision probability of \random is $O(d^2/m)$. In practice, this means that \random should only be used when far less than $\sqrt{m}$ IDs are needed. Concretely, GUIDs are $128$-bit IDs, so they cannot safely handle more than $\sqrt{m} = 2^{64}$ IDs. Unfortunately, with consolidation of cloud computing services and some companies already operating at exabyte scales (i.e., more than $2^{60}$ bytes), we are not so far from a world where we will have enough data objects in a single pool to observe collisions in \random with $128$-bit IDs. Thus, \random is quickly becoming inadequate for large-scale deployments, unless longer IDs are used.

To guarantee collision-freedom of uncoordinated $128$-bit IDs in the long-term future, RocksDB recently started using a different algorithm~\cite{PeterPR1,PeterPR2}, that we call \cluster.\footnote{Notably, MongoDB also uses a variant of \cluster to identify records~\cite{MongoDBObjectID}.} This algorithm picks a random integer $x \in [m]$, and generates IDs sequentially, starting at $x$. Although, experimental data showed that the collision probability of \cluster is significantly lower than that of \random~\cite{PeterPR2}, it is not known exactly how much better \cluster is, or whether there exists an algorithm that is even better than \cluster.

To the best of our knowledge, no other algorithms for the UUIDP are known. In spite of the widespread use of uncoordinated ID-generation algorithms, there has been no theoretical work on this problem. The overly simple setup of the UUIDP may give the false impression that not much can be said about it. In this paper we show that there is more than meets the eye.

\subsection{Our results}

We study the UUIDP in four different settings: using worst-case analysis or competitive analysis; and against an adversary that is oblivious or adaptive.

In the \emph{worst-case} analysis of an algorithm $\calA$, we are interested in minimizing the maximum collision probability of $\calA$, over all adversaries that request at most $d$ IDs in total. (Notice that we must bound the number of requests to make the worst case analysis interesting; otherwise the adversary can drive the collision probability up to $1$.)
In the \emph{competitive} analysis of $\calA$, we are interested, roughly speaking, in minimizing the ratio between the collision probability of $\calA$ and the best possible collision probability. We define these notions more formally in \Cref{sec:model}.

An \emph{oblivious} adversary chooses how many IDs it will request to each one of the $n$ instances before the game begins. The more powerful \emph{adaptive} adversary decides from which instance to request an ID on a step-by-step basis, based on the IDs that were produced previously.

We prove the following results, which we will formally state in \cref{sec:overview}:
\begin{itemize}
    \item There is an algorithm, namely \cluster, that has \emph{worst-case} collision probability $O(nd/m)$ against \emph{oblivious} adversaries that request at most $d$ IDs.
    \item The \emph{worst-case} collision probability of every algorithm against \emph{oblivious} adversaries is $\Omega(nd/m)$.
    \item There is an algorithm that has \emph{worst-case} collision probability $O((nd/m)\log(1 + d/n))$ against \emph{adaptive} adversaries.
    \item There is an algorithm that has \emph{competitive ratio} $O(\log m)$ against \emph{adaptive} adversaries.
    \item The \emph{competitive ratio} of every algorithm against \emph{oblivious} adversaries is $\Omega(\log m)$.
\end{itemize}
\noindent
Naturally, lower bounds against oblivious adversaries also hold against adaptive adversaries, and upper bounds against adaptive adversaries also hold against oblivious adversaries.

Our theorems imply that \cluster is an effective uncoordinated ID-generation mechanism at scales that are orders of magnitude beyond \random's capacity. The $O(nd/m)$ collision probability of \cluster allows it to handle, for example, collections of $d \gg 2^{64}$ objects, using IDs of $128$ bits or even less. The matching collision probability lower bound establishes that no other algorithm can do significantly better.

In the adaptive setting, \cluster is far from optimal (and so is \random). Our new algorithms for the adaptive case are either optimal or almost optimal, and may be of interest to database researchers and practitioners working on systems with stringent security requirements.

\subsection{Related work}

A well-studied problem that deals with ID assignments is the \textit{renaming problem}~\cite{Attiya1990Renaming,Alistarh2011Renaming}. In the renaming problem, there is a distributed system with $n$ processors that communicate asynchronously, and can fail. Each processor starts out with a name (i.e., an ID) from an unbounded space. The goal is to design a distributed algorithm that maps the names into a namespace that is as small as possible, such that no processor is assigned the same name, and using as little communication as possible.

The renaming problem departs from the UUIDP in two crucial ways: communication is allowed, and machines only need to produce a single unique name. In the UUIDP, the case where exactly $1$ ID is requested to each instance is equivalent to a birthday problem~\cite{vonMises1964Selected, McKinney1966BirthdayProblem}.

The existing work on generating unique identifiers in an uncoordinated setting focuses on practical aspects of this problem, like compressing GUIDs~\cite{Lutteroth2008GUIDs}, designing identifiers with security guarantees~\cite{Schaffer2007UUIDs}, and designing identifiers for IoT platforms~\cite{Aftab2020IoTIDs} and mobile environments~\cite{Jesus2006UniqueIDs}. To the best of our knowledge, there is no theoretical work on a problem with similar characteristics to the UUIDP.

\subsection{Organization}

In \Cref{sec:model} we present the mathematical model of the UUIDP, as well as some preliminary definitions that we use in the rest of the paper. In \Cref{sec:overview} we give a full overview of our results and their motivation, including all the algorithms and main theorems. That section presents a complete picture of our work, only leaving out the proofs. The remainder of the paper contains the proofs.
\section{Model and Definitions}
\label{sec:model}

\def\A{\mathcal{A}}
\def\D{\mathcal{D}}
\def\Adv{\mathrm{Adv}}

We fix the size $m \in \mathbb{N}$ of the universe $[m] := \{1,…,m\}$ of IDs. An \defn{ID-generation algorithm} $\A$ generates the IDs from $[m]$ in some randomized order. Equivalently, we can view $\A$ as a distribution on the set of permutations of $[m]$.

Consider a setting with $n$ \defn{instances} of $\A$ that use independent randomness and that do not communicate. A vector $D = (d_1, \dots, d_n) \in [m]^n$, called the \defn{demand profile}, indicates that $d_i$ IDs are requested from the $i$th instance. These requests are made one by one, such that the instance does not know in advance how many IDs it will have to produce. We say a \defn{collision} occurs when the sets of IDs produced by the $n$ instances are \emph{not} pairwise disjoint. The probability that a collision occurs is denoted by $p_\A(D)$.

When $n \geq 2$, we have $p_{\A}(D) > 0$ for any $\A$ and $D$. (This follows from the fact that already for $D = (1, 1)$ we have $p_{\A}(D) > 0$; see \Cref{cor:uniform-lower-bound}.) In other words, collisions cannot be avoided with certainty for $n \geq 2$. Of course, if $n = 1$ we have $p_{\A}(D) = 0$; we call these \defn{trivial} demand profiles.

In our analysis we restrict our attention to natural sets of demand profiles where the number $n$ of instances, the total number $d := ‖D‖₁ = \sum_{i = 1}^n d_i$ of requests or the maximum number $h$ of requests per instance are restricted, such as
\begin{itemize}
    \item $\D_1(n, d) := \{D ∈[m]^n \mid ‖D‖₁ = d\}$,
    \item $\D_1(d) := \bigcup_{i ≥ 2} \D(i, d)$,
    \item $\D_{\infty}^{\leq}(n, h) := \bigcup_{2 \leq i \leq n} [h]^{i}$, and
    \item $\D_{\infty}^{\leq}(h) := \bigcup_{i ≥ 2} [h]^{i}$.
\end{itemize}
$\D_1(n, d)$ and $\D_1(d)$ contain demand profiles with $L^1$-norm \emph{equal to} $d$, whereas $\D_{\infty}^{\leq}(n, h)$ and $\D_{\infty}^{\leq}(h)$ contain demand profiles with $L^{\infty}$-norm \emph{at most} $h$. We will also use $\D_{\infty}^<(h)$ when the $L^{\infty}$-norm is strictly less than $h$.

\paragraph{Worst-case vs. competitive analysis.}
How should we evaluate the performance of an ID generation algorithm $\A$ on a set $\D$ of demand profiles? The simplest option is to consider the \defn{worst-case} collision probability, namely
\[
\max_{D ∈ \D} p_\A(D).
\]
A natural alternative compares for each $D∈\D$ the actual collision probability $p_\A(D)$ to the best possible collision probability for $D$, namely to
\[p_*(D) := \min_{\A'} p_{\A'}(D).%
\footnote{The minimum is well-defined because it is taken over a compact space of distributions over permutations, and $\A' \mapsto p_{\A'}(D)$ is continuous.}{}^,%
\footnote{Our competitive analysis does not compare against an offline setting where the $i$th instance of $\A$ knows in advance how many requests $d_i$ it will receive. This would make $\A$ overly powerful, as it could sometimes avoid collisions with certainty (e.g., when the $d_i$'s are all different).}
\]
When $D$ is a non-trivial demand profile, we have $p_*(D) > 0$. The \defn{competitive ratio} of $\A$ for any non-trivial $D$ is defined as
\[
\frac{p_\A(D)}{p_*(D)}.
\]
The competitive ratio of $\A$ for a set $\D$ of non-trivial demand profiles is then defined as the maximum competitive ratio taken over all $D \in \D$.

\paragraph{Oblivious vs. adaptive adversaries.}

When $D$ is fixed ahead of time, we say that it is produced by an \defn{oblivious adversary}. The collision probability in this case is $p_{\A}(D)$.

We also consider the case where the demand profile $D$ is produced by an \defn{adaptive adversary} $Z$ on the fly. In this case, the total number $n$ of instances of $\A$ that are probed is not known in advance.
Initially, $D = ()$ is empty. When $D = (d_1, \dots, d_i)$, the adversary has three options:
\begin{itemize}
    \item set $D \gets (d_1, \dots, d_i, 1)$, effectively activating a previously dormant instance of $\A$ from which an ID was never requested;
    \item set $D \gets (d_1, \dots, d_j + 1, \dots, d_{i})$ for some $j \in [i]$, effectively requesting another ID from the $j$th instance of $\A$;
    \item stop the game, making $D$ the \defn{final} demand profile.
\end{itemize}
Crucially, at each step the adaptive adversary learns the ID that is produced and his future decisions may depend on this knowledge. The collision probability for this game is denoted by $p_\A(Z)$. Note that the oblivious setting can be thought of as a special case where $Z$ ignores the IDs produced by the instances of $\A$.

Importantly, the adversary knows the algorithm $\A$ it will be playing against. Thus, once $\A$ is fixed, the adversary can tailor its strategy to defeat $\A$.

When evaluating the performance of $\A$ on a set of demand profiles $\D$ in the adaptive setting, we consider the set $\Adv(\D)$ of adversaries that always stop the game eventually, and do so with a final demand profile contained in $\D$. We can again choose between a worst-case analysis that ponders
\[
\max_{Z ∈ \Adv(\D)} p_\A(Z),
\]
and a competitive analysis that considers the competitive ratio of $\A$ for any $Z$ that produces non-trivial demand profiles, defined as
\[\frac{p_\A(Z)}{\displaystyle\E_{D \sim Z}[p_*(D)]},\]
and evaluates the maximum competitive ratio over $Z \in \Adv(\D)$. 
Here we write, in a slight abuse of notation, $D \sim Z$ for the final demand profile $D$ arising randomly from the interaction of $Z$ and~$\A$.

In \cref{tab:evaluation-settings} we summarize the four settings $\{$worst-case, competitive$\} × \{$oblivious, adaptive$\}$ considered in this paper.

\begin{table}[bth]
    \centering
    \begin{tabular}{ccc}
        \toprule
        &worst-case & competitive\\
        \midrule
        oblivious & $\max\limits_{D ∈ \D} p_\A(D)$ & $\max\limits_{D ∈ \D} \frac{p_\A(D)}{p_*(D)}$\\
        adaptive & $\max\limits_{Z ∈ \Adv(\D)} p_\A(Z)$ & $\max\limits_{Z ∈ \Adv(\D)} \frac{p_\A(Z)}{\E_{D \sim Z}[p_*(D)]}$\\
        \bottomrule
    \end{tabular}
    \caption{Four settings for evaluating the performance of an ID generation algorithm $\A$ on a set $\D$ of demand profiles.}
    \label{tab:evaluation-settings}
\end{table}

\paragraph{$O$-notation}
In this paper, we use a non-asymptotic version of $O$-, $\Omega$- and $\Theta$-notation. Specifically, we say that $f(x_1, x_2, \dots) = O(g(x_1, x_2, \dots))$ if there exists a constant $c > 0$ such that $f(x_1, x_2, \dots) = c \cdot g(x_1, x_2, \dots)$ for all $x_1$, $x_2$, $\dots$. In other words, $O$ simply hides constant factors. $\Omega$ and $\Theta$ are defined similarly.
\section{Overview}
\label{sec:overview}

In this section we present the roadmap of the rest of the paper, including all the algorithms we analyze, as well as the main theorems.

\subsection{Basic Results on \random and \cluster}

The starting point of this work is an analysis and comparison of the collision probability of \random and \cluster for arbitrary demand profiles. These two algorithms are defined below. In each case we illustrate how the algorithm might behave. A sequence of $m$ squares represents the ID space $[m]$ and a number $i$ in a square indicates that the corresponding ID is the $i$th ID that is returned.

\tcbset{algorithm/.style={
	colback=black!10!white,left=2mm,right=2mm,width=280pt
}}
\tcbset{text/.style={
	colback=black!10!white,frame empty,top=0pt,bottom=0pt,left=0pt,right=0pt
}}
\tcbset{example/.style={
	colback=white,left=1mm,right=1mm,
	colbacktitle=black!50!white,hbox
}}

\def\exampleScale{0.35}
\def\makeSlots#1{
	\foreach \c[count=\i] in {#1}{
		\ifx\c\empty
			\draw (\i,0) rectangle (\i+1,1);
		\else
			\draw[fill=lightgray] (\i,0) rectangle (\i+1,1);
		\fi

		\node at (\i+0.5,0.5) {\c};
	}
}

\begin{center}
\begin{tcolorbox}[algorithm,title=Algorithm \random]
\begin{tcolorbox}[text]
Return the IDs from $[m]$ in a uniformly random order.
\end{tcolorbox}
\centering
	\begin{tcolorbox}[example,title={Example ($m = 20$, $8$ requests)}]
    	\begin{tikzpicture}[scale=\exampleScale]
    		\makeSlots{,,6,,,2,,,3,4,,,5, ,8,1,7,,,}
    	\end{tikzpicture}
     \end{tcolorbox}
\end{tcolorbox}
\end{center}

\begin{center}
\begin{tcolorbox}[algorithm,title=Algorithm \cluster]
\begin{tcolorbox}[text]
Pick $x ∈ [m]$ uniformly at random and return IDs in the order $x,x+1,x+2,…$, all modulo $m$.
\end{tcolorbox}
\centering
\begin{tcolorbox}[example,title={Example ($m = 20$, $8$ requests)}]
    	\begin{tikzpicture}[scale=\exampleScale]
    		\makeSlots{,,,,,,,,,1,2,3,4,5,6,7,8,,,}
    	\end{tikzpicture}
    \end{tcolorbox}
\end{tcolorbox}
\end{center}

The following theorem gives a tight asymptotic estimate of the collision probability of \cluster.

\begin{restatable}{theorem}{thmcluster}
\label{thm:cluster}
Suppose $n \geq 2$. Let $D ∈ [m]^n$. Then,
\[
p_{\cluster}(D) = \Theta\Big(\min\Big(1, \frac{n\norm{D}_1}{m}\Big)\Big).
\]
\end{restatable}

\noindent
The $\min(1,·)$ safeguard in the probability is a recurring pattern in our bounds.
If it kicks in we have $p_\A(D) = Θ(1)$, which amounts to total failure of $\A$ on $D$. When we know that we are dealing with small collision probabilities, then the $\min(1,·)$ can be stripped away.\nocite{D3}

Instead of analyzing \random directly, we study a more general algorithm that we call $\bins(k)$, which includes $\random$ as a special case for $k = 1$, and which also plays a central role in several lower and upper bounds down the road.

\begin{center}
\begin{tcolorbox}[algorithm,title=Algorithm $\bins(k)$]
\begin{tcolorbox}[text]    
	Partition $[m]$ into $⌊m/k⌋$ \defn{bins} of $k$ IDs and $m \bmod k$ leftover IDs.
    Pick a random permutation of the bins. Iterate over the shuffled bins, returning all IDs of a bin in increasing order before moving on to the next bin. Finally, return the leftover IDs in increasing order.
\end{tcolorbox}
\centering
\begin{tcolorbox}[example,title={Example ($m = 20$, $k = 3$, $8$ requests)}]
    	\begin{tikzpicture}[scale=\exampleScale]
            \begin{scope}[draw=gray]
                \makeSlots{,,,4,5,6,,,,7,8,,1,2,3,,,,,}
            \end{scope}
        	\foreach \i in {0,...,5}{
        		\draw[very thick] (\i*3+1,0) rectangle (\i*3+4,1);
        	}
        \end{tikzpicture}
    \end{tcolorbox}
\end{tcolorbox}
\end{center}

\begin{restatable}{theorem}{thmrandom}
\label{thm:random}
Suppose $n \geq 2$. Let $D ∈ [m]^n$ and $k \in [m]$. Then,
\[
p_{\bins(k)}(D) =
\Theta\Big(\min\Big(1 , \frac{\norm{D}_1^2 - \norm{D}_2^2}{km} + \frac{n \norm{D}_1}{m} + \frac{n^2 k}{m} \Big)\Big).
\]
\end{restatable}
\noindent A calculation shows that, when $k = 1$, the first term in the sum dominates $p_{\bins(k)}(D)$. This yields the collision probability of $\random$.

\begin{corollary}
\label{cor:random}
Suppose $n \geq 2$. Let $D ∈ [m]^n$. Then,
\[
p_{\random}(D) = \Theta\Big(\min\Big(1, \frac{\norm{D}_1^2 - \norm{D}_2^2}{m}\Big)\Big).
\]
\end{corollary}
\noindent
\Cref{thm:random} also implies that $p_{\bins(k)}(D) = \Omega(\min(1, \frac{n\norm{D}_1}{m}))$. Combining this with \Cref{thm:cluster}, we conclude that \cluster is asymptotically no worse than $\bins(k)$ and, consequently, \random.

\begin{corollary}
Suppose $n \geq 2$. Let $D ∈ [m]^n$ and $k \in [m]$. Then,
\begin{align*}
p_{\cluster}(D) &= O(p_{\bins(k)}(D)) \text{ and, in particular,}\\
p_{\cluster}(D) &= O(p_{\random}(D)).
\end{align*}
\end{corollary}

\subsection{Worst-Case Analysis Against Oblivious Adversaries}

To do a worst-case analysis, we need to restrict the total number $d$ of requests. This is because any adversary that requests $m + 1$ or more IDs forces collision probability $1$. We therefore only consider demand profiles from $\D_1(n, d)$. \Cref{thm:cluster} and \Cref{thm:random} imply the following worst-case collision probabilities.

\begin{corollary}
\label{cor:worst-case}
Suppose $n \geq 2$. Let $d \geq n$. Then,
\begin{align*}
    \max_{D ∈ \D_1(n,d)} p_{\cluster}(D) &= \Theta(\min(1, nd/m)) \text{ and }\\
    \max_{D ∈ \D_1(n,d)} p_{\random}(D)  &= \Theta(\min(1, d^2/m)).
\end{align*}
\end{corollary}

Does \cluster have optimal worst-case collision probability? We answer this question positively, proving a matching lower bound. Further, we show that all but a vanishing fraction of demand profiles from $\calD_1(n,d)$ force this worst-case collision probability, implying that \cluster is optimal on almost all demand profiles.

\begin{restatable}{theorem}{thmclustermostlygood}
    \label{thm:worst-case-lower-bound}
    For all but an $\exp(-Θ(n))$-fraction of $D ∈ \D_1(n,d)$
    \[p_*(D) = Ω\Big(\min\Big(1,\frac{nd}{m}\Big)\Big).\]
\end{restatable}

\subsection{Worst-Case Analysis Against Adaptive Adversaries}

Interestingly, when facing adaptive adversaries, the collision probability of \cluster worsens by at least a factor of $n$.
\begin{restatable}{lemma}{lemadaptivedefeatscluster}
\label{lem:adaptive-defeats-cluster}
Suppose $n \geq 2$. Let $d \geq 2n$. There exists an adaptive adversary $Z \in \Adv(\D_1(n,d))$ such that
\[
p_{\cluster}(Z) = \Omega\Big(\min\Big(1, \frac{n^2 d}{m}\Big)\Big).
\]
\end{restatable}

\noindent
Roughly speaking, this adversary exploits the fact that once the initial ID from each instance is known, the two instances with the two closest IDs can be forced to collide if they are at distance less than $d$.

Notice that \random seems to be robust to adaptivity: because every ID produced is a fresh random number, an adaptive adversary can extract little information about future IDs from past IDs. Of course, the collision probability of \random in the adaptive case is as bad as in the oblivious case.

Is there an algorithm for the adaptive setting that matches the collision probability lower bound $\Omega(\min(1, nd/m))$ from \Cref{thm:worst-case-lower-bound}? We give an essentially positive answer, by showing that there exists an algorithm, that we call $\clusterstar$, that is only a small logarithmic factor away from that lower bound.

\begin{center}
\def\run{\mathrm{run}}
\begin{tcolorbox}[algorithm,title=Algorithm \clusterstar]
    \label{algo:clusterstar}
\begin{tcolorbox}[text] 
	Let $\run(x, r)$ be the sequence $(x, x + 1, \dots, x + (r - 1))$ modulo $m$. Repeat the following for $r = 1,2,4,8,…$: Draw $x \in [m]$ uniformly at random, such that $\run(x, r)$ does not collide with previously chosen runs. For the next $r$ requests return the IDs from $\run(x, r)$.
\end{tcolorbox}
\centering
	\begin{tcolorbox}[example,title={Example ($m = 20$, $8$ requests)}]
    	\begin{tikzpicture}[scale=\exampleScale]
        	\makeSlots{,,,,,,4,5,6,7,,1,,,2,3,,,8,}
        \end{tikzpicture}
    \end{tcolorbox}
\end{tcolorbox}
\end{center}

Intuitively, $\clusterstar$ is resilient against adaptive adversaries because an adversary cannot predict too many future IDs from an instance. More specifically, the adversary can only predict a long run of IDs from an instance if it has already requested roughly the same number of IDs from it. Still, $\clusterstar$ behaves similarly to
$\cluster$ because the exponential growth of the runs implies that the number of runs per instance is small, and that most IDs produced by an instance belong to its largest run.

\begin{restatable}{theorem}{thmclusterstar}
\label{thm:adaptive}
Suppose $n \geq 2$. Let $d \geq n$ and $\D := \D_1(d) \cap \D_{\infty}^{\leq}(n, m/(2 \log m))$. Then,
\[ \max_{Z ∈ \Adv(\D)}
p_{\clusterstar}(Z) = O\Big(\min\Big(1, \frac{nd}{m} \log\Big(1 + \frac{d}{n}\Big)\Big)\Big).
\]
\end{restatable}

\noindent
Due to fragmentation of the ID space, $\clusterstar$ may not be able to generate all $m$ IDs according to its rules. This is why we restrict our analysis to demand profiles with at most $m/(2\log m)$ requests per instance. An instance then allocates at most $\log m$ runs of sizes at most $m/(2\log m)$, which always fit even for worst-case fragmentation.

\subsection{Competitive Analysis Against Oblivious Adversaries}

The shortcoming of \cluster is that when the requests are skewed towards a small group of instances, the collision probability is far from optimal.
For instance, let $D = (d - 1, 1) \in \calD_1(2, d)$ be a maximally skewed demand profile with $n = 2$. Consider the following algorithm: the first ID is randomly sampled from $[m - (d - 2)]$, and all other IDs are deterministically taken from $\{m - d + 3, \dots, m\}$. Then, the collision probability on $D$ is $1 / (m-(d-2))$, which is the probability that first IDs of the two instances collide. This is up to a factor $Θ(d)$ smaller than $p_{\cluster}(D) = d/m$.

Therefore, although \cluster is optimal in the worst case, there are algorithms that perform far better on some special cases, i.e., \cluster's competitive ratio is far from optimal. To overcome this limitation, we design yet another algorithm, which we call $\binsstar$.

\begin{center}
\begin{tcolorbox}[algorithm,title=Algorithm $\bins^*$]
\begin{tcolorbox}[text]
 Partition the ID space into $O(\log m)$ chunks and partition the $i$th chunk into bins of $2^{i-1}$ IDs each. Pick a uniformly random bin of size $1$, then of size $2$, then of size $4$, and so on, always returning all IDs of a bin in increasing order before moving on to a bin of twice the size.
\end{tcolorbox}
\centering
	\begin{tcolorbox}[example,title={Example ($m = 32$, $8$ requests)}]
    	\begin{tikzpicture}[xscale=.625*\exampleScale,yscale=\exampleScale]
        	\begin{scope}[draw=gray]
        		\makeSlots{,,,1,,,,,,,2,3,,,,,4,5,6,7,,,,,8,,,,,,}
        	\end{scope}
        	\foreach \i in {0,...,3}{
        		\draw[very thick] (\i*8+1,0) rectangle (\i*8+9,1);
        	}
        	\foreach \j in {1,2,3,4,5,6,7,10,12,14,16,20}{
        		\draw[semithick] (\j+1,0) -- ++ (0,1);
        	}
        \end{tikzpicture}
    \end{tcolorbox}
\end{tcolorbox}
\end{center}

$\bins^*$ combines the exponential allocations from $\clusterstar$ with the spatial partitioning of $\bins(k)$ for $k ∈ \{1,2,4,…\}$. Intuitively, the goal of $\binsstar$ is that instances with similar loads allocate most of their IDs from the same region of the universe $[m]$. In particular, low-demand instances can only collide with few IDs of high-demand instances.

\begin{restatable}{theorem}{thmbinsstarweaklycompetitive}
    \label{thm:bins-star-weakly-competitive}
    The competitive ratio of $\bins^*$ for $\calD_{\infty}^{<}(m / \log m)$ is $O(\log m)$.
\end{restatable}

\noindent
Furthermore, we establish optimality by proving a matching lower bound that even applies to $\D = [\sqrt{m}]²$ (and hence to any $\D'$ with $\D ⊆ \D'$).

\begin{restatable}{theorem}{thmcompetitivelowerbound}
    \label{thm:competitive-lower-bound}
    Every algorithm has competitive ratio $Ω(\log m)$ for $\D = [\sqrt{m}]²$.
\end{restatable}

\subsection{Competitive Analysis Against Adaptive Adversaries}

Finally, we consider the competitive ratio against adaptive adversaries. We prove the following general reduction from oblivious to adaptive adversaries.

\begin{restatable}{theorem}{thmweaktostrong}
    \label{thm:weak-to-strong-competitive}
    Let $\A$ be either $\bins^*$ or $\bins(k)$ for some $k ∈ [m]$, and let $\D$ be some set of demand profiles. If $\A$ has competitive ratio $c$ for $\D$, then $\A$ has competitive ratio at most $4c$ for $\Adv(\D)$.
\end{restatable}

This implies that the $O(\log m)$ upper bound on the competitive ratio of $\bins^*$ from \cref{thm:bins-star-weakly-competitive} also holds in the adaptive case.

\begin{corollary}
    \label{cor:bins-star-strongly-competitive}
    $\bins^*$ is has competitive ratio $O(\log m)$ for $\Adv(\D_{\infty}^{<}(m/\log m))$.
\end{corollary}
\section{Collision Probability On Fixed Demand Profiles}
\label{sec:fixed-profiles}

In this section we study the collision probability of $\bins(k)$ and $\cluster$, on arbitrary demand profiles. We begin with a technical lemma.

\begin{restatable}{lemma}{lemunionboundtight}
\label{lem:union-bound-tight}
Let $\calE_1, \dots, \calE_\ell$ be pairwise independent events. Then,
\[
\Pr\Big[\bigcup_{i = 1}^{\ell} \calE_i\Big] = \Theta\Big(\min\Big(1, \sum_{i = 1}^{\ell} \prob{\calE_i}\Big)\Big).
\]
\end{restatable}
\begin{proof}
The upper bound follows from the union bound. We now prove the lower bound. We consider three cases:
\begin{itemize}
    \item Case 1: $\sum_{i = 1}^{\ell} \prob{\calE_i} \leq 2/3$. By a Bonferroni inequality,
\begin{align*}
\Pr\Big[\bigcup_{i} \calE_i\Big] &\geq \sum_{i} \prob{\calE_i} - \sum_i \sum_{\substack{j < i}} \prob{\calE_i \cap \calE_j}\\
&= \sum_{i} \prob{\calE_i} - \sum_i \sum_{j < i} \prob{\calE_i}\prob{\calE_j} \tag{as the events are pairwise independent}\\
&= \sum_{i} \Big(\prob{\calE_i} \Big(1 - \sum_{j < i} \prob{\calE_j}\Big)\Big)\\
&\geq \Big(1 - \sum_{j} \prob{\calE_j}\Big) \sum_{i} \prob{\calE_i}\\
&\geq \frac 13 \sum_{i} \prob{\calE_i}
= Ω\Big(\min\Big(1, \sum_{i} \prob{\calE_i}\Big)\Big).
\end{align*}
    \item Case 2: $\prob{\calE_j} \geq 1/3$ for some $j$. Then, $\prob{\bigcup_i \calE_i} \geq \prob{\calE_j} = Ω(1) = Ω(\min(1, \sum_i \prob{\calE_i}))$.

    \item Case 3: otherwise. Then, there exists a $j$ such that $\sum_{i = 1}^j \prob{\calE_i} \leq 2/3$ and $\sum_{i = 1}^{j + 1} \prob{\calE_i} > 2/3$. Because $\prob{\calE_{j+1}} < 1/3$, we have
    \[
    1/3 < \sum_{i = 1}^j \prob{\calE_i} \leq 2/3.
    \]
    Thus, applying the first case to the first $j$ events we get
    \[\Pr\Big[\bigcup_{i=1}^ℓ \calE_i\Big] \geq \Pr\Big[\bigcup_{i = 1}^{j} \calE_i\Big] = Ω(1) = Ω\Big(\min\Big(1, \sum_i \prob{\calE_i}\Big)\Big).\qedhere\]
\end{itemize}
\end{proof}

Roughly speaking, \Cref{lem:union-bound-tight} says that the union bound almost gives the right answer when the events involved are pairwise independent.
The only exception is when the union bound in meaningless, i.e., the sum of the probabilities is greater than $1$, then the probability of the union is close to $1$.

We will extensively use the following well-known asymptotic approximations of the exponential function,

\begin{equation}
    1 - \varepsilon \leq e^{-\varepsilon} ≤ 1-\varepsilon/2 \text{ for $ε \in [0,1]$}.
    \label{prp:approximation-exponential}
\end{equation}

We start by establishing the collision probability of \cluster.

\restateThm{thm:cluster}{\thmcluster}
\begin{proof}
Let $D = (d₁,…,d_n)$. 
For $i, j \in [n]$ with $i \neq j$, let $\calC_{ij}$ be the event that there is a collision between the IDs produced by the instances $i$ and $j$.
For every choice of the starting point of instance $j$, there are exactly $d_i + d_j - 1$ choices for the starting point of instance $i$ that cause a collision (this uses $d_i,d_j > 0$). Thus, $\Pr[\calC_{ij}] = (d_i + d_j - 1)/m$.

The sum of the collision probabilities is
\begin{align*}
\sum_{i < j} \prob{\calC_{ij}} &= \sum_{i < j} \frac{d_i + d_j - 1}{m} = \frac{1}{m} \left((n - 1) \norm{D}_1 - \binom{n}{2}\right)
= \Theta\left(\frac{n \norm{D}_1}{m}\right).
\end{align*}

We claim that the collision events $\calC_{ij}$ are pairwise independent. Let $i < j$ and $p < q$ be any two pairs of instances, such that $i \neq p$ or $j \neq q$ (otherwise, they are the same pair). If $|\{i, j, p, q\}| = 4$ then $\calC_{ij}$ and $\calC_{pq}$ are obviously independent. Suppose $|\{i, j, p, q\}| = 3$; without loss of generality, assume $i = p$ and $j \neq q$.
Then, conditioned on $\calC_{ij}$, the probability of $\calC_{iq}$ remains unchanged, because conditioning on $\calC_{ij}$ neither affects the distribution of the IDs of instance $i$, nor the independence of instances $i$ and $q$.

This means that the hypotheses of \Cref{lem:union-bound-tight} are met for the event $\bigcup_{i < j} \calC_{ij}$ that a collision occurs during the game. Hence,
\[
p_{\cluster}(D) = \Theta\Big(\min\Big(1, \sum_{i < j} \prob{\calC_{ij}}\Big)\Big) =  \Theta\Big(\min\Big(1, \frac{n \norm{D}_1}{m}\Big)\Big).\qedhere
\]
\end{proof}

\noindent
For the analysis of $\bins(k)$ we need the following auxiliary lemma.\nocite{D1}

\begin{restatable}{lemma}{lemsampling}
\label{lem:sampling}
Let $\calU$ be a universe of $u$ elements. Let $S_1$ be a set of $s_1$ samples from $\calU$, drawn without replacement. At every step of the sampling process, each of the remaining elements has equal probability of being sampled. Analogously, let $S_2$ be a set of $s_2$ samples from $\calU$, drawn without replacement, independent of $S_1$. Then,
\[
\prob{S_1 \cap S_2 \neq \emptyset} = \Theta\left(\min\left(1, \frac{s_1 s_2}{u}\right)\right).
\]
\end{restatable}
\begin{proof}
    Assume without loss of generality that $1 ≤ s₁ ≤ s₂$ and that $S₂$ is sampled first, after which the set $S₁ = \{x₁,…,x_{s₁}\}$ is sampled element by element. If $s₂>u/4$ then $\Pr[S₁ ∩ S₂ ≠ ∅] ≥ \Pr[x₁ ∈ S₂] ≥ 1/4$ and $\min(1,\frac{s₁s₂}{u}) ≥ \min(1,s₁/4) = Θ(1)$ leaving nothing to show. From now on assume $s₁ ≤ s₂ ≤ u/4$. By the chain rule, we have
    \begin{align*}
        \Pr[S_1 \cap S_2 = ∅] &= \prob{x_1 \notin S_2, \dots, x_{s_1} \notin S_2}\\
        &= \prod_{i = 1}^{s_1} \prob{x_i \notin S_2 \mid x_1 \notin S_2, \dots, x_{i - 1} \notin S_2}\\
        &= \prod_{i = 1}^{s_1} \Big(1 - \frac{s_2}{u - i + 1}\Big).
    \end{align*}
    We can upper-bound this by
    \[
        \prod_{i = 1}^{s_1} \Big(1 - \frac{s_2}{u - i + 1}\Big)
        ≤ \Big(1 - \frac{s_2}{u}\Big)^{s₁}
        \eqrefrel{prp:approximation-exponential}{Eq. }{≤}
        \Big(\e^{-s₂/u} \Big)^{s₁} = \e^{-s₁s₂/u},
    \]
    and lower bound it similarly by
    \[
        \prod_{i = 1}^{s_1} \Big(1 - \frac{s_2}{u - i + 1}\Big)
        ≥ \Big(1 - \frac{s_2}{u/2}\Big)^{s₁}
        \eqrefrel{prp:approximation-exponential}{Eq. }{≥}
        \Big(\e^{-4s₂/u} \Big)^{s₁} = \e^{-4s₁s₂/u}.
    \]
    If $s₁s₂/u \geq 1/100$ we obtain, by combining both bounds,
    \[
        \Pr[S₁ ∩ S₂ ≠ ∅] = 1 - \e^{-Θ(s₁s₂/u)} = Θ(1) = Θ\Big(\min\Big(1,\frac{s₁s₂}{u}\Big)\Big),
    \]
    as desired. Similarly, if $s₁s₂/u < 1/100$, we obtain
    \[
        \Pr[S₁ ∩ S₂ ≠ ∅] = 1 - \e^{-Θ(s₁s₂/u)}
        \eqrefrel{prp:approximation-exponential}{Eq. }{=}
        Θ(s₁s₂/u) = Θ\Big(\min\Big(1,\frac{s₁s₂}{u}\Big)\Big).\qedhere
    \]
\end{proof}

\restateThm{thm:random}{\thmrandom}
\begin{proof}
Let $D = (d₁,…,d_n)$. If any $d_i$ exceeds $⌊m/k⌋·k$ then the theorem correctly predicts a collision probability of $Θ(1)$. Thus, we may assume that no instance runs out of bins. For every $i ∈ [n]$, let $B_i$ be the set of $⌈d_i/k⌉$ bins chosen by the $i$th instance. Notice that $(B_i)_i$ is an independent family, because in our model instances behave independently. Let $C_{ij}$ be the event that instances $i$ and $j$ collide, which happens if and only if there is a bin selected by both instances.
By \Cref{lem:sampling},
\[
\prob{\calC_{ij}} = \prob{B_i \cap B_j \neq \emptyset}\\
= \Theta\Big(\min\Big(1, \frac{\lceil d_i/k \rceil \lceil d_j/k \rceil}{\floor{m/k}}\Big)\Big).
\]

The collision events $\calC_{ij}$ are pairwise independent, by a similar argument as in the proof of \Cref{thm:cluster}. Thus, by \Cref{lem:union-bound-tight},
\begin{align*}
p_{\bins(k)}(D) &= \Theta\Big(\min\Big(1, \sum_{i < j} \prob{\calC_{ij}}\Big)\Big)\\
&= \Theta\Big(\min\Big(1, \sum_{i < j} \min\Big(1, \frac{\lceil d_i/k \rceil \lceil d_j/k \rceil}{\floor{m/k}}\Big)\Big)\Big)\\
&= \Theta\Big(\min\Big(1, \sum_{i < j} \frac{\lceil d_i/k \rceil \lceil d_j/k \rceil}{\floor{m/k}}\Big)\Big)\numberthis\label{eq:p-bins}.
\end{align*}

Now we wish to compute the inner summation. Notice that
\begin{align*}
\frac{\lceil d_i/k \rceil \lceil d_j/k \rceil}{\floor{m/k}}
&= \Theta\left(\frac{(1 + d_i/k)(1 + d_j/k)}{m/k}\right)\\
&= \Theta\left(\frac{k + (d_i + d_j) + d_i d_j/k}{m}\right).
\end{align*}
Straightforward calculations show that
\[
\sum_{i < j} (d_i + d_j) = \Theta(n\norm{D}_1)
\text{\ \ and\ \ }
\sum_{i < j} d_i d_j = \Theta(\norm{D}_1^2 - \norm{D}_2^2).
\]
Thus,
\[
\sum_{i < j} \frac{\lceil d_i/k \rceil \lceil d_j/k \rceil}{\floor{m/k}} = \Theta\bigg(\frac{n^2 k + n \norm{D}_1 + (\norm{D}_1^2 - \norm{D}_2^2)/k}{m}\bigg).
\]
Plugging this equality into \Cref{eq:p-bins} completes the proof.
\end{proof}
\section{Worst-Case Lower Bound Against Oblivious Adversaries}
\label{sec:worst-case}

The goal of this section is to prove \Cref{thm:worst-case-lower-bound}. The proof is in two steps: The first one is to establish a lower bound on the collision probability for the case of \defn{uniform} demand profiles, that is, demand profiles of the form $(h, h, \dots, h)$ for some $h \in \mathbb{N}$. Then, we extend this lower bound to \emph{most} demand profiles, using the fact that most demand profiles ``contain'' a uniform demand profile with almost the same total number of requests.

\subsection{Uniform Demand Profiles}

We will use the following technical lemma.

\begin{restatable}{lemma}{lemballsintobins}
\label{lem:balls-into-bins}
    Assume $n ≥ 2$ balls are thrown independently into $ℓ ≥ n$ bins where the bins are chosen with probabilities $p₁,…,p_ℓ$. The probability that all balls land in distinct bins is maximized if and only if $p₁ = … = p_ℓ = 1/ℓ$.
\end{restatable}
\begin{proof}
For $ℓ = n = 2$ the success probability $2p₁(1-p₁)$ is maximized only for $p₁ = p₂ = 1/2$. Now consider $ℓ ≥ 3$ and assume, for the sake of contradiction, that the maximum is achieved for a non-uniform distribution in which, say, $p₁ ≠ p₂$. We may imagine that the throwing of balls happens in two stages. In the first stage, the balls are thrown into $ℓ-1$ bins with probabilities $p₁+p₂,p₃,…,p_ℓ$. In the second stage, the balls from the first bin are reclaimed and thrown into two new sub-bins with probabilities $p₁/(p₁+p₂)$ and $p₂/(p₁+p₂)$.

Now note that if we vary $p_1$ and $p_2$ while keeping $p₁+p₂$ fixed, the distribution of the balls in the first stage of the experiment is unaffected. And when we equalize the probabilities $p_1$ and $p_2$ of the sub-bins, the success probability in the second stage is improved, by the argument for the case $ℓ = 2$. (Of course, when the number of balls that land in the first bin in the first phase is not equal to $2$, the distribution in the sub-bins does not matter.) Thus, $p_1 = p_2$ yields higher success probability, which is a contradiction.
\end{proof}

\begin{lemma}
    \label{lem:uniform-lower-bound}
    For all $D = (h, \dots, h) \in \mathbb{N}^n$,
    \[
    p_*(D) = p_{\bins(h)}(D).
    \]
\end{lemma}

\begin{proof}
    \def\supp{\mathrm{supp}}
    \def\disj{\mathrm{disj}}
    In the scenario at hand every instance of an ID generation algorithm $\A$ has to produce exactly $h$ IDs, allowing us to think of $\A$ in simpler terms, namely as a distribution $μ_\A$ on $\binom{[m]}{h}$, i.e., a distribution on sets of $h$ IDs. Conversely, any such distribution characterizes an algorithm. Assume $\A$ is optimal for $D$, i.e., $p_\A(D) = p_*(D)$ and that $\A$ minimizes the support size $|\supp(μ_\A)|$ among the optimal algorithms, where
    \[\supp(μ) := \Big\{ X ∈ \binom{[m]}{h} \mid μ(X) > 0 \Big\}.\]
    Our plan is to show that $\A$ behaves like $\bins(h)$, up to relabeling of IDs.
    
    We call a sequence $X₁,…,X_n ∈ \binom{[m]}{h}$ of sets of IDs an \defn{outcome}. When using $\A$ on demand profile $D$ the $n$ instances of $\A$ produce this outcome with probability $μ_\A(X₁) \dots μ_\A(X_n)$. Let 
    \[
    \disj(X₁,…,X_n) =
    \begin{cases}
        1 & \text{if $X_1, \dots, X_n$ are pairwise disjoint}\\
        0 & \text{otherwise}.
    \end{cases}
    .
    \]    
    The probability that \emph{no} collision occurs is $q(μ_\A)$ where $q : \mathbb{R}^{[m] \choose h} \to \mathbb{R}$ is the function
    \begin{equation}
        q(μ) := \sum_{\!\!\!\!{X₁,…,X_n ∈ \binom{[m]}{h}}\!\!\!\!} μ(X₁)\dots μ(X_n) · \disj(X₁,…,X_n).\label{eq:polynomial}
    \end{equation}

\begin{claim}
\label{clm:linear}
    $q(\mu)$ is an affine function\footnote{A function $f : \mathbb{R}^n \to \mathbb{R}$ is \defn{affine} if $f(x_1, \dots, x_n) = a_1 x_1 + \dots + a_n x_n + b$ for some constants $a_1, \dots, a_n, b \in \mathbb{R}$.} on any two entries of $\mu$ corresponding to intersecting sets. More specifically, if $P, Q \in {[m] \choose h}$, with $P \cap Q \neq ∅$, and we fix every entry of $\mu$ except $\mu(P)$ and $\mu(Q)$, the resulting $2$-dimensional function is affine.
\end{claim}
\begin{proof}
    This is because $q(\mu)$ is structurally a polynomial with a variable for every entry $\mu(X)$ of the input $\mu$, and the $\disj(·)$-predicate filters out all monomials that contain $\mu(X)\mu(Y)$ with $X \cap Y \neq ∅$ (including, in particular, $\mu(X)^2$).
\end{proof}
    
    \begin{claim}
        \label{claim:pairwise-disjoint}
        If $P,Q ∈ \supp(μ_\A)$, $P ≠ Q$, then $P ∩ Q = ∅$.
    \end{claim}
    \begin{proof}
        To reach a contradiction, assume $P ∩ Q ≠ ∅$. Then, \Cref{clm:linear} implies that $q(\mu)$ is affine on $P$ and $Q$.
        
        Consider a continuum $(μ_t)_{t ∈ [0,1]}$ of variants of $μ_\A$ where the probability mass on $\{P,Q\}$ linearly shifts from $Q$ to $P$, namely
        \[
            μ_t(X) := \begin{cases}
                (μ_\A(P) + μ_\A(Q))·t & \text{if } X = P\\
                (μ_\A(P) + μ_\A(Q))·(1-t) & \text{if } X = Q\\
                μ_\A(X) & \text {otherwise. }
            \end{cases}.
        \]
        Since $\mu_t$ only linearly varies the entries corresponding to $P$ and $Q$, and $q$ is affine on those entries, the map $t \mapsto q(\mu_t)$ is affine. Note that $μ_\A = μ_t$ for some $t ∈ (0,1)$. Thus, we either have $q(μ₀) ≥ q(μ_\A)$ or $q(μ₁) ≥ q(μ_\A)$. Assume without loss of generality the former. Then $μ₀$ is no worse than $μ_\A$ in terms of success probability and has smaller support than $μ_\A$, contradicting the choice of $\A$.
    \end{proof}
    Hence, $\supp(μ_\A)$ is some collection of pairwise disjoint sets from ${[m] \choose h}$. The sets in $\supp(μ_\A)$ correspond to the disjoint bins used by $\bins(h)$ up to relabeling of IDs. There are just two missing details. First, that $μ_\A$ must assign the same probability to each bin, by \cref{lem:balls-into-bins}.\footnote{An exception is the case $nh > m$. Then collisions are guaranteed and any algorithm is optimal.} Second, that $\supp(μ_\A)$ is a \emph{maximal} set of disjoint sets; otherwise, we can view $\A$ as a variant of $\bins(h)$ that assigns probability $0$ to some bins, which is again sub-optimal by \cref{lem:balls-into-bins}. Overall we obtain
    \[ p_*(D) = p_\A(D) = p_{\bins(h)}(D),\]
    as desired.
\end{proof}

Because $\bins(h)$ is optimal on demand profiles $(h, \dots, h) \in \mathbb{N}^n$, it plays a central role in some of our analyses. In particular, the collision probability on uniform demand profiles, namely
\begin{equation}
\label{eq:balanced-demand-with-binsk}
    p_{\bins(h)}\big((h, \dots, h)\big) = \Theta\Big(\min\Big(1, \frac{n^2 h}{m}\Big)\Big),
\end{equation}
will come in handy. This probability follows directly from \Cref{thm:random}.
\Cref{lem:uniform-lower-bound} and \Cref{eq:balanced-demand-with-binsk} imply the following.

\begin{corollary}
\label{cor:uniform-lower-bound}
Let $D = (h, \dots, h) ∈ ℕ^n$. Then
\[
p_*(D) = \Omega\left(\min\left(1, \frac{n^2 h}{m}\right)\right).
\]
\end{corollary}

\subsection{Proof of Theorem \ref{thm:worst-case-lower-bound}}

The intuition is that most demand profiles are sufficiently similar to the uniform demand profile $(\frac dn,…,\frac dn)$, which we have analyzed in the preceding subsection. Let us call $D ∈ \D_1(n,d)$ \defn{$ε$-good} if at least $εn$ entries of $D$ exceed $εd/n$. Call $D$ \defn{$ε$-bad} otherwise. Let $\B_\varepsilon$ be the set of $\varepsilon$-bad demand profiles in $\D_1(n, d)$.

\begin{restatable}{lemma}{lemmostdemandsbalanced}
    \label{lem:most-demands-balanced}
    There is a constant $ε ∈ (0,\frac 12]$ such that
    \[ |B_ε|/|\D_1(n,d)| = \exp(-Θ(n)).\]
\end{restatable}
\begin{proof}
    In this proof we use standard bounds on binomial coefficients, namely
    \[ \Big(\frac nk\Big)^k ≤ \binom nk ≤ \Big(\frac{n\e}k\Big)^k.\]
    We will also use the following basic combinatorial fact: Let $f(a, b)$ be the number of ways to throw $a$ indistinguishable balls into $b \leq a$ numbered bins, such that every bin contains at least one ball. Then,
\[
f(a, b) = {a - 1 \choose b - 1}.
\]
    
    The total number of demand profiles in $\D_1(n,d)$ is equal to $f(d, n)$, hence
    \[ |\D_1(n,d)| = \binom{d-1}{n-1} ≥ \Big(\frac{d-1}{n-1}\Big)^{n-1} ≥ \Big(\frac dn\Big)^{n-1}. \]
    Each $D ∈ \B_ε$ has at most $εn$ entries exceeding $εd/n$. To upper bound $|\B_ε|$ we first select $εn$ out of $n$ indices for “large” entries that (potentially) exceed $εd/n$. We then specify the remaining $n-εn$ “small” entries explicitly, each between $1$ and $εd/n$. Finally, we specify how the remaining demand, which is at most $d$, is partitioned among the $εn$ large entries. We get
    \begin{align*}
        |\B_ε| &≤ \binom{n}{εn} \Big(\frac{εd}{n}\Big)^{n-εn} f(d, \epsilon n)\\
        &≤ 2^n \Big(\frac{εd}{n}\Big)^{n-1} \Big(\frac{n}{εd}\Big)^{εn-1} \Big(\frac{d\e}{εn-1}\Big)^{εn-1}\\
        &≤ 2^n \Big(\frac{εd}{n}\Big)^{n-1} \Big(\frac{2\e}{ε²}\Big)^{εn-1}.
    \end{align*}
    For the fraction of $ε$-bad demand profiles we obtain, assuming $ε > 0$ is small enough
    \[
        \frac{|\B_ε|}{|\D_1(n,d)|} ≤ 2^n ε^{n-1} \Big(\frac{2\e}{ε²}\Big)^{εn-1} = \exp(-Θ(n)).\qedhere
    \]
\end{proof}

We are now ready to present the proof of the main theorem of this section.
\restateThm{thm:worst-case-lower-bound}{\thmclustermostlygood}
\begin{proof}
    By \cref{lem:most-demands-balanced} it suffices to show that the lower bound on $p_*(D)$ holds for the $ε$-good demand profiles, for a suitable constant $ε$. For any $ε$-good $D ∈ \D_1(n,d)$, we can construct a new demand profile $D' = (\frac{εd}{n},…,\frac{εd}{n})$ with $n' := εn$ entries, by decreasing some entries and removing some other entries from $D$. Then,
    \begin{align*}
        p_*(D) ≥ p_*(D')
        &\stackrel{\text{Lem. \ref{lem:uniform-lower-bound}}}{=}
        p_{\bins(εd/n)}(D')\\
        &\stackrel{\text{Eq. \ref{eq:balanced-demand-with-binsk}}}{=}
        \Omega\Big(\min\Big(1, \frac{\varepsilon n' d}{m}\Big)\Big)\\
        &= \Omega\Big(\min\Big(1, \frac{nd}{m}\Big)\Big).\qedhere
    \end{align*}

\end{proof}

\section{Worst-Case Upper Bound Against Adaptive Adversaries}

In this section we set out to analyze the worst-case collision probability against adaptive adversaries. Our first result shows that \cluster is vulnerable to adaptive adversaries, which can worsen the collision probability (compared to oblivious adversaries) by at least a factor of $n$.

\restateThm{lem:adaptive-defeats-cluster}{\lemadaptivedefeatscluster}
\begin{proof}
Consider an adversary $Z$ that behaves as follows:
\begin{enumerate}
    \item[1.] Request an ID from each of the $n$ instances.
    \item[2.] Pick the two closest IDs; say they were produced by instances $i$ and $j$. Without loss of generality, assume instance $i$ produced the smaller ID of the two.
    \item[3.] Request $d - n$ IDs from instance $i$.
\end{enumerate}

\noindent
Note that $Z ∈ \Adv(\D_1(n,d))$. For each $i \in [n]$, let $x_i$ be the first ID produced by instance $i$. Observe that $Z$ causes a collision if and only if $x_i$ and $x_j$ are at distance at most $d - n - 1$ modulo $m$, for some $i ≠ j$. Let $\calE_i$ be the event that $x_i$ is at distance at least $d - n - 1$ from each of $x₁,…,x_{i-1}$. Then,
\begin{align*}
1 - p_{\cluster}(Z) &= \prob{\calE_1, \dots, \calE_n}\\
&= \prod_{i = 1}^n \prob{\calE_i \mid \calE_1, \dots, \calE_{i - 1}}\\
&\leq \prod_{i = 1}^n \Big(1 - \frac{(i - 1)(d - n)}{m}\Big)\\
&\eqrefrel{prp:approximation-exponential}{Eq. }{≤}
\prod_{i = 1}^n \exp\Big(-\frac{(i - 1)(d - n)}{m}\Big)\\
&= \exp\Big(-\frac{n(n-1)(d - n)}{2m}\Big)\\
&\stackrel{d \geq 2n}{\leq} \exp\Big(-\frac{n(n-1)d}{4m}\Big).
\end{align*}
If $4m < n(n-1)d$, then $1 - p_{\cluster}(Z) \leq e^{-1}$. Thus, $p_{\cluster}(Z) \geq \Omega(1) = \Omega(\min(1, n^2 d/m))$, as desired. Otherwise, $4m \geq n(n-1)d$, so we can apply \Cref{prp:approximation-exponential}, to get
\[
1 - p_{\cluster}(Z) \leq 1 - \frac{n(n-1)d}{8m}.  
\]
Thus, $p_{\cluster}(Z) = \Omega(n^2 d / m) = \Omega(\min(1, n^2 d / m))$, concluding the proof.
\end{proof}

\subsection{A Nearly Optimal Algorithm}

In the rest of this section we show that $\clusterstar$ is optimal in the worst case against adaptive adversaries, up to a logarithmic factor.

\restateThm{thm:adaptive}{\thmclusterstar}
\begin{proof}
We may assume
\begin{equation}
    \label{eq:assumption}
    m = ω(nd\log(1 + d/n));
\end{equation}
otherwise the upper bound holds trivially. Let $Z ∈ \Adv(\D)$.
As soon as an instance selects a run and returns the first ID in it, we say the run has been \defn{opened}. At any point during the game, we say an ID is \defn{active} if it belongs to an opened run of some instance. Notice that the number of active IDs is at most $2d$ at all times, because the number of IDs that are active but have not been requested yet is at most $d$.

Since we aim for an upper bound on the collision probability, we may count any intersection between the open runs of two instances as a collision, regardless of which IDs from these runs have already been returned by the corresponding instances.

Let $\calC_i$ be the event that at least $i$ runs are open and that the $i$th run that is opened collides with an already opened run. Let $I_i ∈ [n]$ be the instance that opens a run and let $T_i := |\{j : j < i, I_j = I_i\}|$ be the number of runs previously opened by the same instance (with $(I_i, T_i) = (⊥, -∞)$ if fewer than $i$ runs are ever opened). Let $\clus$ be an upper bound on the number of runs that can be opened in total.

\begin{claim}
\label{clm:upper-bound-conditioned}
We have for any $t_i ∈ \{-∞,0,1,2,…\}$
\[
\prob{\calC_i \mid \overline{\calC_1}, \dots, \overline{\calC_{i-1}}, T_i = t_i} \leq O\left(\frac{d + \clus 2^{t_i}}{m}\right).
\]
\end{claim}
\begin{proof}
Since $\Pr[C_i \mid T_i = -∞] = 0$ we may assume $t_i ≥ 0$.
When the $i$th run is opened, instance $I_i$ has at least $m - (d + t_i (2^{t_i}-1))$ valid choices for the first ID of the run. This is because the first ID cannot be any of the $d$ or less IDs within the runs $I_i$ has previously opened, nor one of the $2^{t_i} - 1$ IDs directly to the left of the $t_i$ runs that $I_i$ has previously opened. Notice that \Cref{eq:assumption} implies that $m = ω(d\log d)$. Combining this with the fact that $t_i ≤ \log d$, we conclude that there are $Ω(m)$ valid choices.

On the other hand, the number of choices that cause a collision are at most $2d + \clus (2^{t_i}-1)$. This is because a collision occurs when the first ID is either one of the $2d$ or less active IDs, or is one of the $2^{t_i}-1$ IDs to the left of one the $\clus$ or less runs from other instances.

Dividing the choices causing a collision by the number of valid choices yields the stated upper bound.
\end{proof}

\Cref{clm:upper-bound-conditioned} implies
\[ \Pr\big[\calC_i \mid \overline{\calC_1}, \dots, \overline{\calC_{i-1}}, T_i\big] \leq O\Big(\frac{d}{m} + \frac{\clus 2^{T_i}}{m}\Big)\]

Summing over $i$ gives,
\begin{align*}
\sum_{i = 1}^{\clus} \Pr\big[\calC_i \mid \overline{\calC_1}, \dots, \overline{\calC_{i-1}}, T_i\big] &\leq O\Big(\sum_{i = 1}^{\clus} \Big(\frac{d}{m} + \frac{\clus 2^{T_i}}{m}\Big)\Big)\\
&= O\Big(\frac{d\clus}{m} + \frac{\clus}{m} \sum_{i = 1}^{\clus} 2^{T_i}\Big)\\
&= O\Big(\frac{d\clus}{m}\Big), \numberthis\label{eq:conditional-expectation-bound}
\end{align*}
where the last step uses that $\sum_{i = 1}^{\clus} 2^{T_i}$ counts the number of active IDs in the end, which is at most $2d$.
We have,
\begin{align*}
p_{\clusterstar}(Z) &\leq \Pr\big[\calC_1 \cup (\overline{\calC_1} \cap \calC_2) \cup (\overline{\calC_1} \cap \overline{\calC_2} \cap \calC_3) \cup \dots\big]\\
&= \prob{\calC_1} + \Pr\big[\overline{\calC_1} \cap \calC_2\big] + \Pr\big[\overline{\calC_1} \cap \overline{\calC_2} \cap \calC_3\big] + \dots\\
&\leq \sum_{i = 1}^{\clus} \Pr\big[\calC_i \mid \overline{\calC_1}, \dots, \overline{\calC_{i-1}}\big]\\
&= \sum_{i = 1}^{\clus} \E_{T_i}\Big[\Pr\big[\calC_i \mid \overline{\calC_1}, \dots, \overline{\calC_{i - 1}}, T_i\big]\Big] \tag{by the law of total expectation}\\[-5pt]
&= \E_{(T_i)_i}\bigg[\sum_{i = 1}^{\clus} \Pr\big[\calC_i \mid \overline{\calC_1}, \dots, \overline{\calC_{i - 1}}, T_i\big]\bigg] \tag{by linearity of expectation}\\[-5pt]
&\leq O\Big(\frac{d \clus}{m}\Big) \tag{by \Cref{eq:conditional-expectation-bound}}.
\end{align*}

To conclude the proof, we must show that $\clus = O(n \log(1 + d/n))$. Let $d_i$ be the (random) number of IDs requested from instance $i ∈ [n]$. Note that $d_i = 0$ is permitted if only $n' < n$ instances receive at least one request. Instance $i$ creates exactly $\ceiling{\log (1 + d_i)}$ runs. Thus, the total number of runs is
\begin{align*}
    \sum_{i = 1}^n \ceiling{\log(1 + d_i)} &\leq O\Big(\sum_{i = 1}^n \log(1 + d_i)\Big)\\
    &\leq O\Big(n \log\Big(1 + \frac{\sum_{i = 1}^n d_i}{n}\Big)\Big) \tag{by the concavity of $x \mapsto \log(1 + x)$, and Jensen's inequality}\\
    &\leq O(n \log(1 + d/n)). \qedhere
\end{align*}
\end{proof}

\section{Competitive Ratio Upper Bound Against Oblivious Adversaries}

In this section we take a closer look at $\bins^*$, proving that it has competitive ratio $O(\log m)$ for a large class of demand profiles. This result will be strengthened to adaptive adversaries in \Cref{sec:competitive-upper-bound-adaptive}.

\subsection{Missing Details in \texorpdfstring{$\bins^*$}{BINS*}}
\label{sec:bins-star-definition}

Let us start by filling in some missing details from our description of $\bins^*$ from \Cref{sec:overview}. Recall that $\bins^*$ partitions the IDs space $[m]$ into $O(\log m)$ chunks, and these are in turn broken into bins. Specifically, the number of chunks is $C := ⌈\log m - \log \log m⌉$, and each one has $2^{C-1}$ IDs. This works because
\begin{align*}
    C·2^{C-1}
    &= ⌈\log m - \log \log m⌉·2^{⌈\log m - \log \log m⌉ -1}\\
    &≤ \log m · 2^{\log m - \log \log m} = m.
\end{align*}
For each $i \in [C]$, the $i$th chunk is further partitioned into $2^{C-i}$ bins of size $2^{i-1}$ each.

Note that $\bins^*$ does not specify how to proceed after the selected bin from the last chunk has been exhausted. In that case at least $2^C$ IDs have been requested. Since $2^C ≥ m/\log m$, the underlying demand profile would not be in $\D_{\infty}^{<}(m/\log m)$, so \cref{thm:bins-star-weakly-competitive} makes no claim in this case.

\subsection{\texorpdfstring{$\bins^*$}{BINS*} has Competitive Ratio \texorpdfstring{$O(\log m)$}{O(log m)} for Oblivious Adversaries}

We begin by limiting the set of demand profiles we have to consider. For a (non-trivial) demand profile $D$ let $D⁻$ be the \defn{rounded demand profile} arising from $D$ by
\begin{itemize}
    \item first, rounding each entry in $D$ down to the next power of $2$, and
    \item second, if there is a unique largest entry, we reduce it to the second largest entry; when a unique largest entry exists, we call its associated instance \defn{heavy}.
\end{itemize}
For instance if $D = (9,5,4,42)$ then $D⁻ = (8,4,4,8)$. For a set $\D$ of demand profiles we define $\D⁻ := \{D⁻ \mid D ∈ \D\}$.
\begin{lemma}
    \label{lem:round-to-2**i}
    Let $D ∈ \D_{\infty}^{<}(m/\log m)$. Then
    \[ p_{\bins^*}(D) = p_{\bins^*}(D⁻).\]
\end{lemma}
\begin{proof}
    When using $\bins^*$ a collision occurs if and only if two instances allocate the same bin. It therefore only matters which bins are allocated by the instances. An instance allocates its $i$th bin (of size $2^{i-1}$) when the previously allocated bins of sizes $1,2,…,2^{i-2}$ are all full and an additional ID is requested, i.e., when the $2^{i-1}$th ID is requested. Thus, rounding down to powers of $2$ does not change the collision probability.
    
    If after rounding the is a heavy instance with $2^{i-1}$ requests for some $i$, while there are at most $2^{j-1}$ requests for each of the others instances for some $j < i$, then the heavy instance is the only one allocating a bin in each of the chunks $j+1,…,i$. These bins cannot cause a collision, so the heavy instance might as well have only $2^{j-1}$ requests.
\end{proof}

We define the \defn{rank distribution} of a rounded demand profile $D⁻$ to be a vector $(s₁,…,s_k)$ where $s_i$ is the number of times $2^{i-1}$ occurs in $D⁻$ and where $2^{k-1}$ is the largest entry of $D⁻$.

\begin{lemma}
    \label{lem:rounded-demands-lower-bound}
    Let $D⁻$ be a rounded demand profile with rank distribution $(s₁,…,s_k)$. Then
    \[ p_*(D⁻) = Ω\Big(\min\Big(1,\frac 1m \sum_{i = 1}^k \binom{s_i}{2} 2^i\Big)\Big).\]
\end{lemma}
\begin{proof}
    \def\E{\mathcal{E}}
    Assume we use an optimal algorithm $\A$ for $D⁻$, i.e., one with $p_\A(D⁻) = p_*(D⁻)$. For $i \in [k]$, let $\E_i$ be the event that there is a collision between two instances with demand $2^{i-1}$. We have $\Pr[\E_i] = p_\A(D_i)$ where $D_i = (2^{i-1},…,2^{i-1})$ has length $s_i \geq 0$. Recall that $\bins(2^{i-1})$ is the optimal algorithm for $D_i$, by \cref{lem:uniform-lower-bound}. We compute
    \begin{align*}
        \Pr[\E_i] &= p_\A(D_i)\\
        &≥ p_*(D_i)\\
        &= p_{\bins(2^{i-1})}(D_i)\\
        &\eqrefrel{eq:balanced-demand-with-binsk}{Eq.}{=}
        Θ\Big(\min\Big(1,\binom{s_i}{2}\frac{2^{i-1}}{m}\Big)\Big).
    \end{align*}
    Notice that because we use $\binom{s_i}{2}$ instead of $s_i²$, we correctly handle the cases $s_i ∈ \{0,1\}$, using the convention $\binom{0}{2} = \binom{1}{2} = 0$.
    Using that the events $\E₁,…,\E_k$ relate to disjoint sets of instances and are hence independent, we get
    \begin{align*}
        p_*(D⁻) &= p_\A(D⁻)\\
        &≥ \Pr\big[\bigcup_{i=1}^k \E_i \big]\\
        &\refrel{lem:union-bound-tight}{Lem.}{=}
        Θ\Big(\min\Big(1,\sum_{i=1}^k \Pr[\E_i]\Big)\Big)\\
        &= Θ\Big(\min\Big(1,\frac 1m \sum_{i = 1}^k \binom{s_i}{2} 2^i\Big)\Big).\qedhere
    \end{align*}
\end{proof}

We will use the following simple inequalities involving binomial coefficients.
\begin{restatable}{lemma}{lembinomial}
    \label{lem:binomial}
    For any $x,y ≥ 0$
    \begin{enumerate}[(i)]
        • $\binom{x+y}{2} ≤ 3\binom x2 + 2x + \frac 32 \binom y2 + \frac y2$,
        • $\binom{x+y}{2} ≤ 4\binom x2 + \frac 53 \binom y2 + O(1)$.
    \end{enumerate}
\end{restatable}
\begin{proof}
    \emph{(ii)} follows from \emph{(i)} because $2x = o(\binom x2)$ and $\frac y2 = o(\binom y2)$. To prove \emph{(i)} we proceed by induction. For $x = 0$ or $y = 0$ the claim is trivial. For $x,y ≥ 1$ we distinguish two cases. In both we first apply the identity $\binom a2 = \binom{a-1}{2}+\binom{a-1}{1}$ and then the induction hypothesis.
    \begin{itemize}
        \item Case 1: $x ≤ \frac y2$.
        \begin{align*}\textstyle
            \binom{x+y}{2} &\textstyle= \binom{x+(y-1)}{2} + \binom{x+y-1}{1}\\\textstyle
            &\textstyle\reasonrel{Ind.}{≤}
            3\binom x2 + 2x + \frac 32 \binom{y-1}{2} + \frac{y-1}{2} + x+y-1\\\textstyle
            &\textstyle≤ 3\binom x2 + 2x + \frac 32 \binom{y-1}{2} + \frac{y}{2} + \frac y2 + y - \frac 32\\\textstyle
            &\textstyle≤ 3\binom x2 + 2x + \frac 32 \binom{y-1}{2} + \frac{y}{2} + \frac 32\binom{y-1}{1}\\\textstyle
            &\textstyle≤ 3\binom x2 + 2x + \frac 32 \binom{y}{2} + \frac{y}{2}.
        \end{align*}
        \item Case 2: $y ≤ 2x$.
        \begin{align*}\textstyle
            \binom{x+y}{2} &\textstyle= \binom{(x-1)+y}{2} + \binom{x+y-1}{1}\\\textstyle
            &\textstyle\reasonrel{Ind.}{≤}
            3\binom{x-1}2 + 2(x-1) + \frac 32 \binom{y}{2} + \frac{y}{2} + x+y-1\\\textstyle
            &\textstyle≤ 3\binom{x-1}2 + 2x + \frac 32 \binom{y}{2} + \frac{y}{2} + x + 2x-3\\\textstyle
            &\textstyle≤ 3\binom{x-1}2 + 2x + \frac 32 \binom{y}{2} + \frac{y}{2} + 3\binom{x-1}{1}\\\textstyle
            &\textstyle≤ 3\binom{x}2 + 2x + \frac 32 \binom{y}{2} + \frac{y}{2}.\qedhere
        \end{align*}
    \end{itemize}
\end{proof}

\begin{lemma}
    \label{lem:upper-bound-bins*}
    Let $D⁻$ be a rounded demand profile with rank distribution $(s₁,…,s_k)$. Then
    \[p_{\bins^*}(D⁻) = O\bigg(\frac{\log m}{m}\sum_{i=1}^k \binom{s_i}{2} 2^i\bigg).\]
\end{lemma}
\begin{proof}
    A collision within chunk $i$ can only be caused by instances that allocate a bin in chunk $i$. Let $s_{≥ i} := s_i + … + s_k$ be the number of such instances. Any fixed one of such such instances collides with probability $O(\frac{\log m}{m} 2^i)$, because there are $2^{C-i} = Ω(\frac{m}{\log m}2^{-i})$ bins within chunk~$i$. By union bound over all pairs of instances and all chunks we obtain
    \[ p_{\bins^*}(D⁻) ≤ \frac{\log m}{m} \sum_{i = 1}^k \binom{s_{≥i}}{2} 2^i. \]
    We define $X := \sum_{i = 1}^k \binom{s_{≥i}}{2} 2^i$ and compute, with the understanding that $s_{≥ k+1} = 0$,
    \begin{align*}
        X &= \sum_{i = 1}^k \binom{s_{≥i}}{2} 2^i = \sum_{i = 1}^k \binom{s_i + s_{≥i+1}}{2} 2^i\\
        &\stackrel{\text{Lem. \ref{lem:binomial}}}{≤}
        \sum_{i = 1}^k \bigg(4\binom{s_i}{2} + \frac 53 \binom{s_{≥i+1}}{2} + O(1)\bigg)·2^i\\
        &= 4\bigg(\sum_{i = 1}^k \binom{s_i}{2}2^i\bigg) + \frac 56 \bigg(\sum_{i = 1}^k \binom{s_{≥i+1}}{2}2^{i+1}\bigg) + O(1)·2^{k+1}\\
        &≤ O\Big(\sum_{i = 1}^k \binom{s_i}{2}2^i\Big) + \frac 56·X.
    \end{align*}
    In the last step we used that $s_{k}≥2$ so that the $O(1)·2^{k+1}$ term is dominated by the last term of the first sum. After subtracting $\frac 56 X$ we see that $X = O(\sum_{i = 1}^k \binom{s_i}{2} 2^i)$, which implies the claim.
\end{proof}

\restateThm{thm:bins-star-weakly-competitive}{\thmbinsstarweaklycompetitive}

\begin{proof}
    Let $D ∈ \D_{\infty}^{<}(m/\log m)$ and $(s₁,…,s_k)$ the rank distribution of $D⁻$. Then
    \begin{align*}
        p_{\bins^*}(D)
        &\stackrel{\text{Lem. \ref{lem:round-to-2**i}}}{=}
        p_{\bins^*}(D⁻)\\
        &\stackrel{\text{Lem. \ref{lem:upper-bound-bins*}}}{=}
        O\Big(\min\big(1,\frac{\log m}{m}\sum_{i=1}^k \binom{s_i}{2} 2^i\big)\Big)\\
        &≤O(\log m) \min\big(1,\frac{1}{m}\sum_{i=1}^k \binom{s_i}{2} 2^i\big)\\
        &\stackrel{\text{Lem. \ref{lem:rounded-demands-lower-bound}}}{≤}
        O(\log m) p_*(D⁻) ≤ O(\log m) p_*(D).\qedhere
    \end{align*}
\end{proof}
\section{Competitive Ratio Lower Bound Against Oblivious Adversaries}

In the following we consider the performance of an algorithm $\A$ on demand profiles of the form $D = (i,j)$ for $i,j ∈ [m]$, which will serve as “hard instances” for our lower bounds down the line. Denote $\A(i)$ the distribution of the set of the first $i$ IDs returned by an instance of $\A$. Moreover, let $q_{c,i} := \Pr_{X \sim \A(i)}[c ∈ X]$ for $c,i ∈ [m]$.

\begin{lemma}
    \label{lem:lower-bound-two-demands}
    Let $\A$ be any algorithm and let $i,j ∈ [m]$. We have
    \[ p_\A\big((i,j)\big) ≥ \frac{1}{\min(i,j)}\sum_{c ∈ [m]} q_{c,i}q_{c,j}. \]
\end{lemma}
\begin{proof}
    We will use that a bounded random variable $Z ∈ \{0,1,…,B\}$ satisfies $\Pr[Z > 0] = \Pr[Z ≥ 1] ≥ \E[Z]/B$. We have
    \begin{align*}
        p_\A\big((i,j)\big) &= \Pr_{\substack{X \sim \A(i)\\Y \sim \A(j)}}[|X ∩ Y| > 0]\\
        &≥ \frac{1}{\min(i,j)} \E_{\substack{X \sim \A(i)\\Y \sim \A(j)}}[|X ∩ Y|]\\
        &= \frac{1}{\min(i,j)} \sum_{c ∈ [m]} \Pr_{\substack{X \sim \A(i)\\Y\sim \A(j)}}[c ∈ X ∩ Y]\\
        &= \frac{1}{\min(i,j)} \sum_{c ∈ [m]} q_{c,i}q_{c,j}.\qedhere
    \end{align*}
\end{proof}

The following lemma determines the best achievable collision probability for demand profiles $D = (i,j)$, up to constants.
\begin{lemma}
    \label{lem:collision-prob-two-demands}
    Let $1 ≤ i ≤ j ≤ m/2$. Then, $p_*\big((i,j)\big) = Θ(\frac{i}{m})$.
\end{lemma}

\begin{proof}
    Using the optimality of $\bins(i)$ on $D = (i,i)$ we get
    \[
        p_*\big((i,j)\big) ≥ p_*\big((i,i)\big)
        \refrel{lem:uniform-lower-bound}{Lem.}{=}
        p_{\bins(i)}\big((i,i)\big)
        \eqrefrel{eq:balanced-demand-with-binsk}{Eq. }{=}
        Ω(\tfrac{i}{m}).
    \]
    For the upper bound we construct an algorithm $\A$ for the fixed demand profile $(i, j)$. The algorithm sets aside $j-i$ hard-wired IDs. The first $i$ requests are handled using $\bins(i)$ on the rest of the IDs. All other requests (which are at most $j-i$) are served from the hard-wired IDs.

    A collision on the demand profile $(i, j)$ occurs if and only if there is a collision between the first $i$ IDs of each instance, produced via $\bins(i)$ on the reduced ID space with $m - j + i$ IDs. This probability is $p_{\bins(i)}\big((i,i)\big) = \Theta(1/(m-j+i)) = \Theta(i/m)$.
\end{proof}

Let $k = ⌊\frac 12 \log(m)⌋$. Consider the following distribution $Φ$ on demand profiles:
\begin{equation}
    \label{eq:phi}
    \Pr_{D \sim Φ}[D = (2^i,2^j)] = \frac{1}{W} · 2^{-\max(i,j)} \text{ for $0 \leq i, j \leq k$,}
\end{equation}
where $W$ is a normalization factor. Note that
\[W := \sum_{i = 0}^k \sum_{j=0}^k 2^{-\max(i,j)} ≤ 2 \sum_{i ≥ 0} \sum_{j≥i} 2^{-j} = 2 \sum_{i ≥ 0} 2^{-i+1} = 8 = O(1).\]

\begin{lemma}
    \label{lem:competitiveness-lower-bound}
    For any algorithm $\A$ and $Φ$ as above we have
    \[\E_{D \sim Φ}[p_\A(D)] = Ω\Big(\frac{\log²m}{m}\Big).\]
\end{lemma}

\begin{proof}
    We reuse the numbers $q_{c,i}$ from \cref{lem:lower-bound-two-demands}. We have
    \begin{align*}
        \E_{D \sim Φ}[p_\A(D)]
        &=\frac{1}{W}\sum_{i = 0}^k \sum_{j = 0}^k 2^{-\max(i,j)} p_\A\big((2^i,2^j)\big)\\
        &\refrel{lem:lower-bound-two-demands}{Lem.}{≥}
        \frac{1}{W}\sum_{i = 0}^k \sum_{j = 0}^k 2^{-\max(i,j)} \frac{\sum_{c ∈ [m]} q_{c,2^i}q_{c,2^j}}{\min(2^i,2^j)}\\
        &≥ \frac{1}{W}\sum_{i = 0}^k \sum_{j = 0}^k 2^{-(i+j)} \sum_{c ∈ [m]} q_{c,2^i}q_{c,2^j}\\
        &=\frac{1}{W}\sum_{c ∈ [m]} \sum_{i = 0}^k 2^{-i} q_{c,2^i} \sum_{j = 0}^k 2^{-j} q_{c,2^j}\\
        &=\frac{1}{W}\sum_{c ∈ [m]} \bigg(\sum_{i = 0}^k 2^{-i} q_{c,2^i}\bigg)²
        ≥\frac{1}{Wm} \bigg(\sum_{c ∈ [m]} \sum_{i = 0}^k 2^{-i} q_{c,2^i}\bigg)²\\
        &≥\frac{1}{Wm} \Big(\sum_{i = 0}^k 2^{-i} 2^i \Big)²
        =\frac{1}{Wm} Ω(\log² m) = Ω\Big(\frac{\log²m}{m}\Big).\qedhere
    \end{align*}
\end{proof}

\restateThm{thm:competitive-lower-bound}{\thmcompetitivelowerbound}

\begin{proof}
    Let $\A$ be an algorithm with competitive ratio $c$ for $[\sqrt{m}]²$. Note that $[\sqrt{m}]²$ includes the support of the distribution $Φ$ from \cref{eq:phi}. We have
    \begin{align*}
        \qquad\E_{D \sim Φ}[p_\A(D)]
        &= \frac{1}{W} \sum_{i = 0}^k \sum_{j = 0}^k 2^{-\max(i,j)} p_\A\big((2^i,2^j)\big)\\
        &≤ \frac{c}{W} \sum_{i = 0}^k \sum_{j = 0}^k 2^{-\max(i,j)} p_*\big((2^i,2^j)\big)\\
        &≤ \frac{2c}{W} \sum_{i = 0}^k \sum_{j = i}^k 2^{-j} p_*\big((2^i,2^j)\big)\\
        &\refrel{lem:collision-prob-two-demands}{Lem.}{≤}
        \frac{2c}{W} \sum_{i = 0}^k \sum_{j = i}^k 2^{-j} · O\Big(\frac{2^{i}}{m}\Big)\\
        &≤ O\Big(\frac{c}{m} \sum_{i = 0}^k 2^i \sum_{j ≥ i} 2^{-j}\Big)\\
        &= O\Big(\frac{c}{m} \sum_{i = 0}^k 2^i · 2^{-i+1}\Big) = O\Big(\frac{c \log m}{m}\Big).
    \end{align*}
    Since $\E_{D \sim Φ}[p_\A(D)] = Ω(\frac{\log²m}{m})$ by \cref{lem:competitiveness-lower-bound}, we conclude that $c = Ω(\log m)$, as claimed.
\end{proof}

\section{Competitive Ratio Upper Bound Against Adaptive Adversaries}
\label{sec:competitive-upper-bound-adaptive}

Assume we are using an ID generation algorithm $\A$, the current demand profile is $D = (d₁,…,d_n)$ and no collision has occurred yet. In \cref{lem:adaptive-defeats-cluster} we have seen an example where an adaptive adversary of $\cluster$ observes and exploits useful information about the state of the game namely, which two instances of $\cluster$ produced initial IDs that are close to one another. In contrast, $\random$, $\bins(k)$ and $\bins^*$ exhibit a symmetry between the bins that makes every game state with the same current demand profile (and no collision so far) equivalent. This means that an adaptive adversary cannot make any meaningful observations that could guide its choices, except for its ability to stop immediately when a collision occurs. This suggests that the competitive ratio for these algorithms does not increase (much) if adaptive adversaries are considered. 
In this section, we will prove the following formalization of this intuition:

\restateThm{thm:weak-to-strong-competitive}{\thmweaktostrong}

\def\adv{\mathrm{fol}}

\noindent
The following notions will be useful. A sequence $S = (D₀,D₁, …, D_k)$ of demand profiles is called a \defn{$\D$-demand sequence} if $D₀ = ()$, $D_k ∈ \D$ and $D_{i+1}$ arises from $D_{i}$ either by appending a $1$ or by incrementing one of the entries of $D_i$.
For such an $S$ there is a simple adaptive adversary $\adv(S)$ that follows the demand sequence $S$ as long as no collision has occurred. Should a collision occur when the current demand profile is $D_i$, $\adv(S)$ transitions to a demand profile $\tilde{D}_i ∈ \D$ reachable from $D_i$ that minimizes $p_*(\tilde{D}_i)$, and then immediately stops the game.\footnote{If $\D$ is downward closed, then we always have $D_i ∈ \D$ and hence $\tilde{D}_i = D_i$, meaning $\adv(S)$ always stops immediately upon a collision.} We call $\adv(S)$ a \defn{semi-adaptive} adversary, as its only adaptive decisions depend on whether or not a collision has occurred so far.

In the following we take the adversarial perspective trying to \emph{maximize} the competitive ratio of $Z ∈ \Adv(\D)$ against a fixed $\calA$.

\begin{proof}[Proof of \cref{thm:weak-to-strong-competitive}.]
We fix the algorithm $\calA$ to be either $\bins^*$ or $\bins(k)$ and an arbitrary adaptive adversary $Z \in \Adv(\D)$ with competitive ratio $c'$. We have to show $c' ≤ 4c$.

Below we will prove the following three claims:

\begin{claim}
\label{clm:adaptive-1}
    There exists an adaptive adversary $Z' \in \Adv(\D)$ that achieves competitive ratio $c'$, and that only learns after every step whether a collision has occurred or not. In particular, $Z'$ does not learn the IDs produced by the instances of $\A$.
\end{claim}

\begin{claim}
\label{clm:adaptive-2}
    There exists a $\D$-demand sequence $S$ such that the semi-adaptive adversary $\adv(S)$ achieves competitive ratio at least $c'$.    
\end{claim}

\begin{claim}
\label{clm:adaptive-3}
    There exists a demand profile $D ∈ \D$ (i.e., an oblivious adversary) that achieves competitive ratio at least $c'/4$.
\end{claim}

    The last claim immediately implies what we need: since $\A$ has competitive ratio $c$ for $\D$, the competitive ratio $c'/4$ for $D ∈ \D$ from the claim satisfies $c'/4 ≤ c$.
\end{proof}
We now prove the claims in order.

\begin{proof}[Proof of \Cref{clm:adaptive-1}.]
    Consider the case of $\A = \bins(k)$. For $\A = \bins^*$ a similar argument works.

    Since $\A$ repeatedly selects bins of size $k$ and then produces the IDs within the selected bins in increasing order, the $i$th ID selected by an instance of $\A$ only carries information if $i \equiv 1 \pmod{k}$, and this information is precisely the identity of the bin that has been selected. Let $Z''$ be a variant of $Z$ that is only told these bin identities in a correspondingly modified game. Clearly $Z''$ has the same competitive ratio as $Z$.

    Now consider a further modification of the game where a uniformly random permutation $π$ of the bins is picked beforehand and where any bin identity that would be given to $Z''$ is first permuted by $π$. Since $\bins(k)$ is symmetric under any permutation of the bins and since bin permutations do not affect whether or not a collision has occurred, the random process as observed by $Z''$ does not change and the competitive ratio is still unaffected. Furthermore, the only relevant information obtained by $Z''$ is whether or not a collision has occured:
    Conditioned on a collision having occurred, no further information on the produced IDs makes a difference.
    Conditioned on \emph{no} collision having occurred, the sequence of distinct bin identities that is given to $Z''$ is stochastically independently of the behavior of the instances of $\A$.
    
    We can therefore construct a version $Z'$ of $Z''$ that only learns if a collision has occurred so far and produces a random sequence of distinct bin identities internally, if needed.
\end{proof}

\begin{proof}[Proof of \Cref{clm:adaptive-2}.]
    The adversary $Z'$ from \Cref{clm:adaptive-1} may use randomization in every step. However, it learns no information during the game (except whether a collision has occurred) and can select the demand sequence $S$ that it follows (as long as no collision has occurred) at the start of the game, possibly at random according to some distribution $Δ$. If a collision does occur, then the course of action that maximizes the competitive ratio is to stop as soon as possible with a final demand profile $D ∈ \D$ minimizing $p_*(D)$. In other words, we may assume that $Z'$ behaves like $\adv(S)$ for $S \sim Δ$. Then,
    \begin{align*}
        c' &= \frac{p_\A(Z')}{\displaystyle\E_{D \sim Z'}[p_*(D)]}\\
        &= \frac{\displaystyle\E_{S \sim Δ}[p_\A(\adv(S))]}{\displaystyle\E_{S \sim Δ}\Big[\displaystyle\E_{D \sim \adv(S)}[p_*(D)]\Big]}\\
        &≤ \max_S \frac{p_\A(\adv(S))}{\displaystyle\E_{D \sim \adv(S)}[p_*(D)]}.
    \end{align*}
    In the last step we used that $\frac{a}{b} ≤ \frac{c}{d}$ implies $\frac{a+c}{b+d} ≤ \frac{c}{d}$ for any $a,b,c,d ∈ ℝ^+$ and, more generally,
    \[
        \frac{\sum_{i ∈ [n]} a_i}{\sum_{i ∈ [n]} b_i}
        ≤ \max_{i ∈ [n]} \frac{a_i}{b_i}
        \qquad
        \text{for $a_i, b_i ∈ ℝ^+$ and $i ∈ [n]$.}
    \]
    The demand sequence $S$ maximizing the quotient yields a semi-adaptive adversary $\adv(S)$ with competitive ratio at least $c'$, as desired.
\end{proof}

\begin{proof}[Proof of \Cref{clm:adaptive-3}.]
    Let $S = (D₀,…,D_k)$ be the $\D$-demand sequence from \Cref{clm:adaptive-2}. Note that $p_\A(\adv(S)) = p_\A(D_k)$, that the probability that $\adv(S)$ reaches the last demand profile without causing a collision is $\Pr_{D \sim \adv(S)}[D = D_k] = 1 - p_\A(D_{k-1})$, and that $p_\A(D_i)$ is increasing in $i$. There are two cases:
    \begin{itemize}
        \item Case 1: $p_\A(D_k) ≤ \frac 12$.
        Using that $\adv(S)$ has competitive ratio at least $c'$ we get
        \begin{align*}
            c' &≤ \frac{p_\A(\adv(S))}{\displaystyle\E_{D \sim \adv(S)}[p_*(D)]}\\
            &≤ \frac{p_\A(D_k)}{\displaystyle\Pr_{D \sim \adv(S)}[D = D_k]·p_*(D_k)}\\
            &≤ \frac{p_\A(D_k)}{(1-p_\A(D_{k-1})) · p_*(D_k)}\\
            &≤ \frac{p_\A(D_k)}{(1-p_\A(D_{k})) · p_*(D_k)}\\
            &≤ \frac{p_\A(D_k)}{\frac 12 · p_*(D_k)}
            = 2·\frac{p_\A(D_k)}{p_*(D_k)}.
        \end{align*}
        Thus, $D_k$ has competitive ratio at least $c'/2$.
        \item Case 2: $p_\A(D_k) ≥ \frac 12$.
        Since $p_\A(D_0) = 0$, there exists an index $i$ such that $p_\A(D_{i-1}) ≤ \frac 12$ and $p_\A(D_i) ≥ \frac 12$. Let $J ∈ [k]$ be the random index of the step in which the first collision occurs, if any, and $J = k$ otherwise. Then,
        \begin{align*}
            c' &≤ \frac{p_\A(\adv(S))}{\displaystyle\E_{D \sim \adv(S)}[p_*(D)]}\\
            &≤ \frac{2·p_\A(D_i)}{\sum_{j = 1}^k \Pr[J = j] \cdot p_*(\tilde{D}_j)}\\
            &≤ \frac{2·p_\A(D_i)}{\displaystyle \Pr[J ≥ i] \cdot p_*(D_i)}\\
            &≤ \frac{2·p_\A(D_i)}{(1-p_\A(D_{i-1}))·p_*(D_i)}\\
            &≤ \frac{2·p_\A(D_i)}{\frac 12·p_*(D_i)}
            = 4· \frac{p_\A(D_i)}{p_*(D_i)}.
        \end{align*}
        Thus, $D_i$ has competitive ratio at least $c'/4$.\qedhere
    \end{itemize}
\end{proof}

\section*{Acknowledgments}

We thank the anonymous reviewers of PODS '23 for their valuable comments that greatly improved the presentation of our paper. We gratefully acknowledge support from NSF grants CNS-2118620 and CCF-2106999, as well as Deutsche Forschungsgemeinschaft (DFG) grant 465963632.

\bibliographystyle{plain}
\bibliography{bibliography}

\end{document}